\documentclass[11pt]{article}
\input{style.sty}

\newcommand{\PRP}{\mathrm{PRP}}

\begin{document}

\title{Spectrum Estimation is Almost as Hard as Tomography}
\author{
 Marco Fanizza\thanks{Inria, T\'el\'ecom Paris--LTCI, Institut Polytechnique de Paris, Palaiseau, France. Email: \href{mailto:marco.fanizza@inria.fr}{marco.fanizza@inria.fr}.}
\and
Ryan O'Donnell\thanks{Computer Science Department, Carnegie Mellon University. Email: \href{mailto:odonnell@cs.cmu.edu}{odonnell@cs.cmu.edu}.}
\and
Chirag Wadhwa\thanks{School of Informatics, University of Edinburgh. Email: \href{mailto:chirag.wadhwa@ed.ac.uk}{chirag.wadhwa@ed.ac.uk}.}
}
\date{}
\maketitle

\begin{abstract}
    We study the sample complexity of estimating and testing fundamental unitarily invariant properties of unknown quantum states; namely, the tasks of spectrum estimation, von Neumann entropy estimation, and rank-testing. For $d$-dimensional states, and for every $\gamma>0$, we prove a sample complexity lower bound of $\Omega(d^{2-\gamma})$ for spectrum estimation to constant sorted total-variation error, entropy estimation to constant additive error, and rank-testing to constant trace distance. Our hard instances are constructed from sandwiched products of Haar-random projectors, suitably normalized using a novel technique that lets us derive explicit expressions for high-order tensor moments of the resultant states. For our instances, these moments can be expressed as symmetric functions of Jucys--Murphy elements of the symmetric group algebra. To show that two such mixtures are indistinguishable, we analyze the log-likelihood ratio and perform moment-matching, i.e., we set its low-order Jucys--Murphy components to zero. Indistinguishability is then obtained by bounding an $f$-divergence through the high-order components; the separated high-order terms and concentration of functions of Haar-random unitaries also imply separations in typical spectra, entropies, and ranks, proving all our lower bounds.
\end{abstract}

\newpage

\section{Introduction}
\label{sec:intro}

The spectrum captures the intrinsic properties of a quantum state, invariant under changes of basis. Among the functions of the spectrum, the von Neumann entropy plays a central role: it characterizes the optimal asymptotic compression rate of i.i.d.\ quantum sources~\cite{schumacher1995quantum}, even without prior knowledge of the source~\cite{hayashi2002quantum}, and functionals built from von Neumann entropies determine optimal asymptotic performance of many quantum information processing tasks~\cite{wilde2013quantum}. For example, for bipartite pure states, the entropy of a reduced state equals the optimal asymptotic rate of entanglement concentration and the entanglement cost~\cite{bennett1995concentrating}.
The spectrum also provides information on quantum correlations. For bipartite pure states, the eigenvalues of reduced density operators completely characterize deterministic entanglement transformations ~\cite{nielsen1999conditions}. In multipartite systems, marginal spectra constrain possible entanglement classes~\cite{walter2013entanglement}, and share important connections with representation-theoretic quantities~\cite{christandl2006spectra,klyachko2004quantum}. From a physical perspective, entanglement entropy and the entanglement spectrum have become important tools for characterizing quantum many-body systems and phases of matter~\cite{amico2008entanglement,li2008entanglement}.

These foundational roles of spectra and von Neumann entropies make the task of estimating them an important primitive. Although optimal measurements and asymptotic concentration properties of a natural empirical spectrum estimator have been known for a long time~\cite{keyl2001estimating, hayashi2002quantum, christandl2006spectra}, the precise fundamental limits of spectrum and entropy estimation have remained a challenging problem. In fact, learning the spectrum is clearly an easier task than learning the full density matrix, 
and algorithms for learning quantum states have typically proceeded in two stages: first, estimate the state's eigenvalues, and then use this information to learn its eigenvectors \cite{keyl2006quantum,haah2016sample,pelecanos2026debiased}. While describing a $d$-dimensional mixed state requires $\Theta(d^2)$ parameters, it only has $d$ eigenvalues; one may thus expect the first stage of such algorithms to be far more efficient than the second. This reasoning also applies even more strongly to entropy estimation, since the entropy is a scalar. Contrary to this expectation, the standard ``Empirical Young Diagram'' (EYD) algorithm used in such two-stage procedures is actually known to require $\Omega(d^2)$ copies for both spectrum estimation~\cite{o2021quantum} and entropy estimation (as entropy of the estimated spectrum)~\cite{acharya2020estimating}, which are also sufficient for tomography. 

Despite this limitation, recent work has developed sophisticated algorithms tailored to spectrum estimation that succeed with asymptotically fewer copies than those needed for full state tomography \cite{pelecanos2026beating,pelecanos2026keyl}, at least in the constant-precision regime. However, these algorithms are only known to provide polylogarithmic savings over learning the entire state. Moreover, the only known lower bound for spectrum estimation is $\Omega(d)$, and is implied by that for the much simpler task of mixedness testing \cite{childs2007weak,o2021quantum}. While \cite{pelecanos2026beating} provided compelling numerical evidence that $d^{2 - o(1)}$ copies of a state are necessary for spectrum estimation, proving any superlinear lower bound has long remained open, as highlighted by \cite{wright2016learn,pelecanos2026keyl}. Thus, we aim to address the following concrete question:

\begin{center}
    \emph{Does learning a state's eigenvalues require nearly as many copies as learning it fully?}
\end{center}

Our central contribution is to prove this conjectured $d^{2 - o(1)}$ lower bound for constant-precision spectrum estimation, answering this long-standing open question in the affirmative. This result also provides a near-optimal characterization of spectrum estimation in the regime of constant precision, partially answering an open question of \cite{anshu2024survey,pelecanos2026debiased}.

Before stating our results, let us formally define the task of spectrum estimation. In the sequel, $\mathrm{d_{TV}^\downarrow}$ denotes the total variation distance between sorted distributions (see \Cref{def:sorted-total-variation} for a formal definition).

\begin{definition}
    \label{def:spectrum-estimation}
    Given $n$ copies of an unknown state $\rho \in \mathbb{C}^{d \times d}$, we say an algorithm estimates its spectrum to within precision $\eps$ if it outputs a vector $\widehat{\balpha} \in \mathbb{C}^d$ such that, with probability at least $2/3$, 
    \begin{equation}
        \dtvsorted{\widehat{\balpha}}{\spec(\rho)} \leq \eps.
    \end{equation}
\end{definition}

Here, $\spec(\rho)$ is the spectrum of $\rho$, i.e., a list of its eigenvalues. We show the following lower bound for estimating this quantity:

\begin{theorem}
\label{thm:intro-spec}
    For every even $k \geq 2$, there exists a dimension-independent $\eps_k > 0$ such that the copy complexity of spectrum estimation with precision $\eps_k$ is at least $\Omega(d^{2 - \frac{2}{k/2+2}})$. In other words, at least $d^{2 - o(1)}$ copies are necessary for constant-precision spectrum estimation of $d$-dimensional states.
\end{theorem}

Together with the $O(d^2 \cdot (\log \log d/ \log d)^2)$ upper bound of \cite{pelecanos2026keyl}, our result presents a near-optimal dimension dependence for the task of spectrum estimation. It also implies that the EYD algorithm, which succeeds at spectrum estimation with $\Theta(d^2)$ copies, is near-optimal for this task.

We prove our lower bounds by showing the hardness of distinguishing between a certain pair of mixtures of quantum states. This pair is chosen such that a pair of randomly drawn states from these mixtures will have distant spectra. We also show that such states are separated in their \emph{von Neumann entropies}, i.e., $S(\rho) \coloneqq -\Tr(\rho \ln(\rho))$. This entropy separation and the former indistinguishability result immediately imply the following lower bound.

\begin{theorem}
    \label{thm:intro-entropy}
    For every even $k \geq 2$, there exists a dimension-independent $\eps_k > 0$ such that the copy complexity of estimating a $d$-dimensional state's von Neumann entropy to within precision $\eps_k$ is at least $\Omega(d^{2 - \frac{2}{k/2+2}})$.
\end{theorem}
To our knowledge, the prior best-known lower bound for entropy estimation was $\Omega(d/\log d)$, implied by the classical lower bounds for Shannon entropy estimation \cite{valiant2011estimating,wu2016minimax}. We have provided a near-quadratic improvement by showing that $d^{2 - o(1)}$ copies of a state are also necessary to estimate its von Neumann entropy. Together with the bounds of \cite{acharya2020estimating}, our result shows that the EYD achieves a near-optimal performance for this task.

We also show that the pairs of ensembles considered here have separated ranks, implying a near-quadratic lower bound for testing the rank of a state.

\begin{theorem} 
    \label{thm:intro-rank}
    For all $\gamma > 0$ and $0 < \beta < 1$, there exists a dimension-independent $\eps_{\gamma,\beta} > 0$ such that testing whether a state has rank at most $\beta \cdot d$ or is $\eps_{\gamma,\beta}$-far in trace distance from all such states requires $\Omega(d^{2 - \gamma})$ copies of the state.
\end{theorem}

Our result provides a quadratic improvement over the prior best-known lower bound for this task implied by a result of O'Donnell and Wright \cite{o2021quantum}, who showed that $\Omega(r)$ copies are necessary for testing whether a state has rank $r$. Further, \cite{o2021quantum} showed that the more stringent task of rank-testing with \emph{one-sided} error\footnote{Here, testing with one-sided error requires the tester to always accept when the state has rank at most $r$.} can be solved with $\Theta(r^2)$ copies. Our lower bound is thus near-optimal for this task, and demonstrates that the problems of rank-testing with one-sided and two-sided error are almost equally hard.

\subsection{Technical Overview}
\label{sec:tech-overview}

To prove our lower bounds, we will construct a sequence of unitarily-invariant testing problems, one for each even $k \geq 2$, such that the complexity of each test is $\Omega(d^{2 - \frac{4}{k+4}})$. Prior to our work, the best-known lower bounds for testing any unitarily invariant property of a state were for state certification, and specifically, its special case of mixedness testing. These lower bounds are typically proven by showing the hardness of distinguishing the maximally mixed state from a state drawn at random from a suitable mixture (see e.g., \cite{bubeck2020entanglement,chen2022tight,odonnell2026instance,wadhwa2026optimal}). However, state certification \emph{upper} bounds \cite{buadescu2019quantum} imply that $O(d)$ copies are sufficient for any such point-versus-mixture test. Thus, to prove superlinear lower bounds, we must shift away from this paradigm and consider harder testing problems.

In particular, we will consider distinguishing between two mixtures, labeled as $a$ and $b$. These tasks are of the following form: one either receives a random state $\brho_a$ or $\brho_b$, sampled from different laws, and must determine which of the two laws the unknown state came from. Proving such mixture-versus-mixture lower bounds is notoriously hard, and to our knowledge, no superlinear testing bounds have appeared in the quantum information literature for such problems. Nevertheless, we are motivated to consider mixture-versus-mixture tasks due to their wide success in lower bounds for estimating symmetric properties of \emph{classical distributions}, including for the analogous tasks of sorted distribution estimation \cite{han2018local} and Shannon entropy estimation \cite{valiant2011estimating,wu2016minimax}. While such distinguishing tasks have proven fruitful for distribution testing lower bounds, proving these lower bounds required several technical developments, some of which do not immediately generalize to our quantum setting. 

We start by outlining these ideas from classical distribution testing in \Cref{sec:classical-techniques}. In \Cref{sec:quantum-mixtures}, we use these ideas to motivate our choices of random states $\brho_a$ or $\brho_b$ that we will consider in this work. Lastly, in \Cref{sec:Jucys--Murphy-moment-matching}, we outline our lower bound proofs for distinguishing between such mixtures of states.

\subsubsection{Classical mixture-versus-mixture lower bounds}
\label{sec:classical-techniques}

To prove lower bounds for testing symmetric properties of distributions, a typical strategy is to construct pairs of distributions with matched moments, with the intuition that if two distributions $p,q$ have matched moments $\sum_{i = 1}^d p_i^j = \sum_{i = 1}^d q_i^j$ up to some large degree $j \leq k-1$, then distinguishing between them requires estimating $k$-wise collision probabilities. However, for balanced distributions, such collisions are only likely to be seen when the number of samples is roughly $n \gtrsim d^{1 - \frac{1}{k}}$.

The above observation encourages us to construct pairs of distributions that have as many matched moments as possible while preserving other desirable properties, e.g., separated entropies. A typical strategy for this is to sample each entry $\bp_i$ appropriately at random, such that (a)~$\bp_i$ and $\bq_i$ have matched moments, (b)~$\Ex[\bp_i] = \frac1d$, ensuring that the resulting vector is close to being normalized with high probability. However, one must formally argue that such unnormalized vectors closely represent the behaviour of actual distributions. While a few different strategies exist for this, the most prominent one is Poissonization. 

The method of Poissonization suggests that the number of samples itself can be randomized, i.e., drawn from the distribution $\mathrm{Poisson}(n)$. It is well-known that the sample complexity in this model is equivalent to that in the standard model up to constant factors (see e.g.,~\cite{valiant2008testing,jiao2015minimax}). While this Poissonized access model may appear strange, it can be helpful for the following reason: this perspective is equivalent to imagining one independently receives $\mathrm{Poisson}(n\bp_i)$ samples of each entry $i \in [d]$. Consequently, it suffices to bound the statistical distance between the distributions $\mathrm{Poisson}(n\bp_i)$ and $\mathrm{Poisson}(n\bq_i)$. In other words, Poissonization transforms a mixture-versus-mixture task with $n$ samples to one with a \emph{single} sample, greatly simplifying the analysis. Moreover, this is a completely valid setting even if the randomly sampled $\bp_i$s don't exactly add up to 1, completely bypassing the normalization issue. The technique of Poissonization has been widely applicable in distribution testing, and has yielded tight lower bounds for testing many symmetric properties~\cite{jiao2015minimax,valiant2011estimating,han2018local,wu2016minimax}.

\subsubsection{Mixtures of quantum states}
\label{sec:quantum-mixtures}

Moving to the quantum setting, the most natural idea for candidate hard instances is to take unitarily-invariant state ensembles with far enough, fixed spectra. It is not hard to guess pairs of separated spectra with low-order moment-matching, and indeed some examples were considered by~\cite{pelecanos2026beating} to present striking numerical evidence of a superlinear lower bound. However, it is unclear how to prove that the resulting ensembles are indistinguishable, i.e., that $\Dtr{\E\rho_{a}^{\otimes n}}{\E\rho_{b}^{\otimes n}}$ is small if $n=o(d^{2-\gamma})$, for some $\gamma<1$. 

Generalizing arguments from the classical case, one could instead imagine independently randomizing each eigenvalue of a quantum state to aid with moment-matching, followed by the application of a Haar-random unitary to make the ensembles as hard to distinguish as possible. More generally, one could pick any well-studied unitarily invariant ensemble of matrices to construct the hard mixtures of quantum states. However, having the random matrices obtained in either case suitably normalized and ensuring that the resulting mixtures of states are tractable to analyze seem incredibly difficult to achieve simultaneously. It does not appear that drawing the number of copies of a state from a Poisson distribution would help with either of these concerns in the quantum case. 

A key contribution of this work is a novel idea that gets around both of these issues simultaneously. As we will discuss in depth later, many unitarily invariant ensembles of positive semidefinite matrices have well-characterized $n$-fold tensor moments, i.e., for a random matrix $\bX$ drawn from certain ensembles, one can derive explicit closed-form expressions for $\Ex [\bX^{\ot n}]$. While a typical strategy for drawing random states would be to draw $\bX$ and then produce the state $\brho = \bX / \Tr(\bX)$, $n$-fold moments of such states have the form
\begin{equation}
\label{eq:proportionality}
   \Ex \left[
        \frac{\bX^{\otimes n}}{\Tr(\bX^{\otimes n})}
    \right],
\end{equation}
which is much harder to characterize. To make use of the elegant tensor moment formulas available, we instead introduce an $n$-dependent tilt to our initial distributions. In particular, we draw a random matrix $\wt{\bX}$ with density $\Tr(\bX)^n / \Ex[\Tr(\bX)^n]$ relative to that of $\bX$. This has the effect of producing a global normalization, such that the state ${\brho} = \wt{\bX} / \Tr( \wt{\bX})$ satisfies
\begin{equation}
    \Ex {\brho}^{\otimes n} = \frac{\Ex[\bX^{\otimes n}]}{\Ex [\Tr(\bX)^n]}.
\end{equation}
By decoupling the global normalization factor from the random matrix $\bX$, we can exploit explicit expressions for $\Ex[\bX^{\otimes n}]$ while ensuring that we produce valid quantum states. Going forward, we will discuss initial ensembles of $\bX$ without this tilt, but imagine that the state was indeed drawn with respect to the tilted law.

Let us now describe the kind of distributions we will draw $\bX$ from. In general, the operator $\Ex \bX^{\otimes n}$, for $\bX$ with a unitarily invariant law, commutes with representations of both the unitary group and the symmetric group. Via Schur-Weyl duality~\cite{goodman2009symmetry}, the commutant of $\{U^{\otimes n}:U\in\mathrm U(d)\}$ is the algebra of permutations on $(\mathbb{C}^{d})^{\otimes n}$, and the unitary twirl projects onto it~\cite{CS06}.
This implies that $\E \bX^{\otimes n}$ lies in its center, generated by symmetric polynomials in the (commuting) \emph{Jucys--Murphy elements}~\cite{JUCYS1974107}, i.e.
 
 \begin{equation}
    J_t \coloneqq \sum_{1 \leq i < t} (i\ t),
\end{equation}
where $(i\ t)$ is a transposition. We will often find it convenient to write a normalized form of these elements, i.e., $\wt{J}_t \coloneqq J_t/d.$

In light of the above observation, our central idea is to use ensembles with particularly simple explicit formulas for $\E \bX^{\otimes n}$ in terms of the Jucys--Murphy elements, and exploit the fact that expressing the trace distance or other distinguishability measures in terms of Jucys--Murphy makes it feasible to advance the calculation with simple combinatorics, rather than careful control of cancellations in expressions involving Schur polynomials~\cite{o2021quantum}. Several classical random-matrix ensembles admit explicit formulas of this kind, and behind the scenes these expressions are related to generalizations of Selberg integrals~\cite{forrester2008importance,Kadell1993}. We find it sufficient to build on the simplest examples, which already have a notable use in quantum information theory~\cite{KarolZyczkowski_2001, nechita2007asymptotics, collins2016random}, and they let us avoid using representation theory explicitly. In exchange for this simplicity, the pairs of spectra we consider are no longer deterministic, but we still manage to show that they are typically separated using concentration inequalities for functions of Haar-random unitaries~\cite{Meckes_2019}.

\subsubsection{ Matching moments of Jucys--Murphy elements}
\label{sec:Jucys--Murphy-moment-matching}

As discussed in the previous section, we can draw random states $\brho_a,\brho_b$ satisfying \Cref{eq:proportionality}. Let $\ol{\rho}_a^{(n)},\ol{\rho}_b^{(n)}$ denote $\Ex \brho_a^{\ot n}, \Ex \brho_b^{\ot n}$ respectively. To avoid handling the decoupled normalization factors explicitly, we consider the log-likelihood ratios on a common subspace where their eigenvalues are positive:
\begin{equation}
    \log(\rho_b^{(n)}) -  \log(\rho_a^{(n)}) = \lambda_0 \mathbbm{1} + \log(\Ex \bX_b^{ \otimes n}) - \log(\Ex \bX_a^{\otimes n}),
\end{equation}
where $\lambda_0$ consists of the contributions due to these normalization factors. As we show in \Cref{lem:triangular-upper}, these scalar terms can be entirely neglected when proving statistical indistinguishability. 

Now, taking inspiration from classical moment-matching, a natural strategy is to consider mixtures that allow us to approximate this log-likelihood ratio as a symmetric polynomial in Jucys--Murphy elements whose first non-zero term (barring the scalar) has the highest degree possible. In other words, we aim to perform moment-matching at the level of Jucys--Murphy elements. 

As a warmup, we consider distinguishing between the cases where a) $\bX_a$ is a random projector of rank $d/2$, and b) $\bX_b$ is obtained by tracing out a $d$-dimensional subsystem from a $d^2$-dimensional Haar-random state\footnote{Of course, here, $\bX_b$ is already a quantum state and no normalization is necessary.}. These ensembles are natural candidates to begin our analysis with: every state from the first ensemble has a deterministic spectrum (and thus also entropy and rank), and the asymptotic spectral and entropic properties of the second ensemble have been studied in detail in quantum information theory~\cite{nechita2007asymptotics,page1993average, vivo2016random,wei2017proof,bianchi2019typical}; we further prove that these spectral properties of states from the second ensemble are well-separated from those of the first with high probability (see \Cref{lem:warmup-separation}). Thus, showing that these ensembles are indistinguishable is sufficient to obtain new lower bounds for spectrum and entropy estimation as well as rank testing.

Now, towards indistinguishability, we show that for these particular ensembles,
\begin{equation}
    \rho_a^{(n)} \propto \prod_{t = 1}^n \frac{1+2 \wt{J}_t}{1 + \wt{J}_t}, \qquad \rho_b^{(n)} \propto \prod_{t = 1}^n (1 + \wt{J}_t);
\end{equation}
see \Cref{prop:the-key} and \Cref{lem:pagestates} for exact statements and proofs. Using the specific form of the log-likelihood ratio thus obtained, we show that its first two moments cancel out and the relative entropy depends only on the third moments of Jucys--Murphy elements:
\begin{equation}
    D(\rho_a^{(n)} \| \rho_b^{(n)}) \leq 2 \sum_{t = 1}^n \Tr(\rho_a^{(n)} \wt{J}_t^3) \leq C\frac{n^3}{d^4},
\end{equation}
implying $n \geq \Omega(d^{4/3})$ to distinguish between these ensembles.

For the more general lower bounds, we will aim to match far more than $2$ moments, with the intuition that this should lead to stronger lower bounds. The only degree of freedom above was in the rank of the random projector considered. Matching more moments will require many more parameters, and so we will instead consider products of many random matrices. In particular, we will use sandwiched products of random projectors of suitably chosen ranks~\cite{collins2005product}. Suppose we wish to match $k-1$ moments for some even $k \geq 2$. Picking an appropriate $K = K(k)$, we choose suitable parameters $a,b \in \mathbb{N}^K$. For $e \in \{a,b\}$, we multiply $K$ random projectors of ranks $d/e_1, \dots, d/e_K$, and output the random matrix
\begin{equation}
    \bX_e = \bPi_1 \dots \bPi_{K-1} \bPi_K \bPi_{K-1} \dots \bPi_1,
\end{equation}
where the sandwiching ensures the resulting operator is Hermitian and positive semidefinite, ensuring that we sample genuine random states after normalization. 

The key property we use of these ensembles is the following explicit expression for their $n$-fold tensor moments:
\begin{equation}
    \Ex \bX_e^{\otimes n} \propto \prod_{t = 1}^n f_e(\wt{J}_t), \quad \textnormal{where } \quad f_e(z) = \prod_{i = 1}^K \frac{1+e_i z}{1+z},
\end{equation}
and similarly for the parameters $b$. The log-ratio of these functions is remarkably well-behaved, with the following expression:

\begin{equation}
    h(z) \coloneqq \log \frac{f_b(z)}{f_a(z)} = \sum_{j = 1}^\infty \frac{(-1)^{j+1}}{j} \parens*{\sum_{i = 1}^K b_i^j - \sum_{i = 1}^Ka_i^j} z^j.
\end{equation}
Thus, to remove low-degree moments from the log-likelihood ratio, we need to find two vectors of $K$ integers with matching moments. This is the well-known Prouhet--Tarry--Escott problem \cite{wright1959prouhet,borwein2002prouhet}, and it is known that one can match $k-1$ moments for $K = 2^{k-1}$ (see \Cref{prop:prouhet}).

We are thus able to find choices of $a,b$ for which we can write
\begin{equation}
\label{eq:h-approximation}
    h(z) = \beta_k z^k + \beta_{k+1} z^{k+1} + O(z^{k+2}),
\end{equation}
for sufficiently small $z$. Consequently, the log-likelihood ratio is expressible in the desirable form
\begin{equation}
\label{eq:log-likelihood-good-1}
    \log\ol{\rho}_b^{(n)} - \log\ol{\rho}_a^{(n)} = \lambda_0 \mathbbm{1} + \beta_k \sum_{t = 1}^n \wt{J}_t^k + \beta_{k+1} \sum_{t = 1}^n \wt{J}_t^{k+1} + R, 
\end{equation}
with
\begin{equation}
\label{eq:log-likelihood-good-2}
    |R|\preceq C \sum_{t = 1}^{n} \wt{J}_t^{k+2};
\end{equation}
at least, the above holds on the space corresponding to small eigenvalues of the $\wt{J}_t$s. As \Cref{eq:h-approximation} only holds for small $z$, we cannot apply it to arbitrary eigenvalues of $\wt{J}_t$.

Consequently, even though $\log\ol{\rho}_b^{(n)}$ and $\log\ol{\rho}_a^{(n)}$ commute and their divergences can be computed as classical divergences and in terms of the eigenvalues of Jucys--Murphy elements, one cannot directly resort to a uniform bound on the log-likelihood. 
To handle this obstacle, one can carve out the contributions to such divergences from the low-probability space where eigenvalues of $\wt{J}_t$ are large, and apply the log-likelihood ratio bound only on the remaining space where these eigenvalues are bounded. Then, the contribution from the well-conditioned eigenvalues can safely be bounded using the form in \Cref{eq:log-likelihood-good-1,eq:log-likelihood-good-2} and from bounds on high-order moments of Jucys--Murphy elements which we prove in \Cref{sec:indistinguishability}. Among the various $f$-divergences, we found the triangular discrimination (equivalent to the Hellinger squared and the Jensen-Shannon divergences up to constant factors) to give the cleanest path to an upper bound along these lines. These arguments finally lead to the $\Omega(d^{2 - \frac{4}{k+4}})$ lower bounds for distinguishing between these mixtures.

Lastly, we outline the only remaining step of our proofs, which is to show that these ensembles typically have separated spectra, entropies, and ranks. The simplest and elementary part of the argument is that if the (normalized) power sums $\mathrm{pow}_j(\alpha), \mathrm{pow}_j(\beta),$ of two distributions $\alpha,\beta$ match up to degree $L-1$, and $\alpha_i,\beta_i\leq B$, then the sorted total variation can be bounded by the $L$th moment mismatch

\begin{equation}
    \label{eq:spectral-separation-via-power-sums}
    \dtvsorted{\alpha}{\beta} \geq \frac{1}{2LB^{L-1}} \cdot \abs*{\mathrm{pow}_L(\alpha)-\mathrm{pow}_L(\beta)}.
\end{equation}

However, for random ensembles, the power sums are also random variables. In fact, since the ranks of the projectors in the construction are fixed, they are functions of independent Haar-random unitaries. 
 A fundamental result on the concentration of Lipschitz functions of independent Haar-random unitaries can be used to say that moment-matching holds with high probability. Specifically, we use~\cite[Theorem~5.17]{Meckes_2019}:
\begin{equation}
\Pr(|F(\bU_1,\cdots, \bU_K)-\E F(\bU_1,\cdots, \bU_K)|\geq u)\leq 2e^{-\frac{(d-2)u^2}{24\Lambda^2}}\,,
\end{equation}
for $\Lambda$-Lipschitz function $F$. 

Then, the strategy to obtain concentration of the power sums crucially uses two key facts:

\begin{itemize}
\item The average normalized trace stays bounded: 
\begin{equation}
        \Ex \tr(\bX_e)= \nu_e, \qquad \nu_e=\Theta_k(1)\,, 
\end{equation}
which is a peculiar property of our engineered ensembles.
\item Probabilities of events under the tilted law can be bounded in terms of the untilted one:
\begin{equation}
\Pr(\wt{\bX}_e\in A)\leq \nu_e^{-n}\Pr({\bX}_e\in A)\,,
\end{equation}
so that exponential decay of $\Pr({\bX}_e\in A)$ can win over $\nu_e^{-n}$\,.
\end{itemize}
Once these two ingredients are in place, it is easy to see that the spectral power sums of $\brho_{e}$ can be controlled by the Lipschitz functions $F_j=\tr(\bX_e^j)$ in the high-probability region. There, they concentrate around the moment proxies
\begin{equation}
\Psi_{e,j}
\coloneqq
\frac{
\E\!\left[\tr(\bX_e^j)\right]
}{
\E[\tr(\bX_e)]^j
}\,,
\end{equation}
which approximately match up to degree $k$ by construction. With~\Cref{eq:spectral-separation-via-power-sums}, we now obtain the high-probability separation in spectra, yielding our spectrum estimation lower bound.

To prove a high-probability entropy separation, we cannot directly work with the instances above, as they may have some incredibly small eigenvalues which make it impossible to handle this quantity. To get around this issue, we obtain new hard instances by applying the depolarizing channel to the above ones. By the data processing inequality, these new mixtures remain indistinguishable. Having used the depolarizing channel to ensure the eigenvalues are well-conditioned, we can now express the entropy difference as a series of gaps of power sums, which can be lower bounded by looking at the first non-matched degree. This entropy separation and the statistical indistinguishability now yield our entropy estimation lower bound.

Lastly, to prove rank separation, we handle the two mixtures separately. Recall that our random matrices are products of random projectors of fixed ranks. For the mixture constructed with the lowest-rank projector, we argue that the resultant matrix's rank matches that of this projector almost surely; call this rank $r$.  For the other mixture, we show that there is a sufficiently large set of eigenvalues that are smaller than the largest $r$ eigenvalues, yet are likely to lie above a desirable threshold. The distance of the resulting state from any state of rank $r$ can then be lower bounded by the total spectral mass of this set, giving us our desired rank separation, and consequently our rank testing lower bound.

\subsection{Related Work}

\paragraph{Symmetric properties of probability distributions:} Our work is closely related to analogous classical problems in distribution testing. For the problem of estimating a sorted distribution, the complexity is known to be $\Theta(d / \epsilon^2 \log d)$ for sufficiently large $\epsilon$, and $\Theta(d/\epsilon^2)$ otherwise~\cite{valiant2011estimating,han2018local}. The complexity of estimating a distribution's Shannon entropy was shown to be $\Theta(\max\{\frac{d}{\epsilon \log d}, \frac{\log^2d}{\epsilon^2}\})$~\cite{valiant2011estimating,wu2016minimax,jiao2015minimax}, and that of estimating its support size is $\Theta(\frac{d}{\log d} \log^2(1/\eps))$. As alluded to previously, these lower bounds for symmetric property testing and estimation rely on \emph{moment matching}, which can be traced back to~\cite{raskhodnikova2009strong,valiant2008testing}.

\paragraph{Related results in quantum property testing:} For the problems considered in our work, the representation-theoretic analysis singles out weak Schur sampling as the information-theoretic optimal measurement. Its output is the so-called empirical Young diagram (EYD), which is an asymptotically correct estimator of the spectrum, once normalized. Despite the measurement being information-theoretically optimal, the EYD itself may need to be classically post-processed to obtain optimal rates. However, most prior results analyzed the performance of \emph{plug-in estimators} that directly evaluate target spectral properties on the EYD. In particular, \cite{o2021quantum} test a state's rank using the empirical Young diagram, and \cite{o2021quantum,acharya2020estimating} showed that $\Theta(d^2)$ copies are necessary and sufficient for the plug-in estimator to succeed at spectrum and entropy estimation. We improve on these lower bounds by proving nearly matching ones against \emph{all} possible algorithms. We also note that for spectrum estimation with constant precision, \cite{pelecanos2026keyl} recently developed a new algorithm that succeeds with $o(d^2)$ copies, outperforming the EYD estimator and beating the complexity of tomography. Previously, \cite{pelecanos2026beating} showed that unentangled measurements on $o(d^3)$ copies are sufficient for this, again beating the corresponding tomography rate.

\paragraph{Concurrent and independent work:} Shortly after we posted the first version of our paper on arXiv, two new papers \cite{lee2026sample,wang2026unified} were posted as well. \cite{lee2026sample} considers the tasks of Uhlmann fidelity and spectrum estimation; for the latter, they prove a slightly stronger $\Omega(d^2/ \mathrm{polylog}(d))$ copy complexity lower bound. Their proofs make use of a similar normalizing tilt, but they analyze divergences between Schur-Weyl distributions instead of directly bounding divergences between the state mixtures as we have done here. \cite{wang2026unified} proves even stronger lower bounds, in particular obtaining an $\Omega(\frac{d^2}{\epsilon^2 \log^2d })$ spectrum estimation lower bound for $\epsilon \leq O(1/\log d)$. More generally, \cite{wang2026unified} provides a``quantum moment-matching" master theorem, which can be applied to obtain lower bounds for estimating any unitarily invariant property. The first version of \cite{wang2026unified} also handled indistinguishability via a representation-theoretic analysis, but the latest version proves these lower bounds at the level of operators.

\subsection{Outlook}

 In the classical field of distribution testing, moment-matching and Poissonization have been crucial for proving mixture-versus-mixture lower bounds, yielding tight lower bounds for a wide array of problems, including sorted distribution estimation, entropy estimation, and support size estimation \cite{valiant2008testing,valiant2011estimating,han2018local,jiao2015minimax,wu2016minimax}. By drawing distributions from tilted densities as considered in this work, one can also recover classical moment-matching-based lower bounds, providing an alternative normalization technique that generalizes well to the quantum setting. Further, we have simplified the task of proving mixture-vs-mixture lower bounds by considering distributions of operators whose $n$-fold moments have explicit and convenient forms as functions of Jucys--Murphy elements. We hope these techniques will further inspire the development of new lower bounds for learning and testing quantum states.

 While most results in quantum learning and testing were initially obtained through the use of representation theory of the unitary and symmetric groups \cite{o2021quantum,odonnell2016efficient,haah2016sample,pelecanos2026debiased}, a slew of recent work has recovered these results, and in many cases proven new ones, without or with minimal use of this machinery \cite{buadescu2019quantum,odonnell2026instance,pelecanos2025mixed}. Our work adds to this growing body of literature, and we hope that the new techniques developed here will aid in the development of more lower bounds obtained by directly proving the indistinguishability of states via quantum divergences. 

 Specifically for spectrum estimation, we have considered lower bounds against the most general class of algorithms that can make fully entangled measurements. Practical considerations motivate the study of more restricted algorithms, such as those that can only perform unentangled measurements. We expect the use of such tilted distributions to aid with lower bounds for spectrum estimation even against algorithms that can only perform these weaker operations. We remark that the only known bounds for spectrum estimation in such settings are the unentangled $\Omega(d^{3/2})$-lower bounds implied by state certification \cite{chen2022tight}, and the $O(d^3 \cdot (\log \log d / \log d)^4)$ upper bounds of \cite{pelecanos2026beating}, leaving a large gap to be addressed.

 While we prove near-optimal lower bounds for spectrum estimation for constant $\epsilon$, getting the tight $\epsilon$-dependence also remains open. While the upper bound of \cite{pelecanos2026keyl} has a $1/\eps^4$-dependence, they conjectured that the correct dependence is $\Theta\parens*{\frac{d^2}{\epsilon^2\log^2 d}}$, at least for $\epsilon$ not too small. Close to the completion of this work, and with access to an earlier version of our manuscript, ChatGPT 5.6 produced a candidate proof of the conjectured $\Omega\parens*{\frac{d^2}{\epsilon^2\log^2 d}}$ lower bound. We have made this LLM-generated proof available \href{https://github.com/chirag-w/spectrum-estimation/blob/main/near-optimal-spectrum-estimation.pdf}{online}, and will disseminate a simplified proof in the near future. We note that an  $\Omega(\frac{d^{4/3}}{\epsilon^{2/3}})$ bound for $\epsilon\gtrsim d^{-1/4}$ can also be obtained by a slight modification of the warmup hard instance, taking a random projector of rank $r=\frac{td}{t+1}$ and Haar random states with auxiliary register of size $k=td$, and tuning $t=\Theta\parens*{\frac{1}{\epsilon^2}}$. The proof goes along the same separation and moment-matching ideas of the manuscript, and we do not report it.
 
\subsection{Organization}

We present preliminaries in \Cref{sec:prelim} and relevant results on random projections in \Cref{sec:projections}. We then present our warmup $\Omega(d^{4/3})$ lower bounds in \Cref{sec:warmup}. We present the full details of our hard mixtures in \Cref{sec:contruction}, along with the skeleton of our lower bound proofs. Statistical indistinguishability of these instances is shown in \Cref{sec:indistinguishability}, with the separations in spectra, entropies, and ranks shown in \Cref{sec:spectral-separation,sec:entropy-separation,sec:rank-separation} respectively.

\subsection*{Acknowledgments}

M.F.\ thanks T.C.\ Fraser and Harold Nieuwboer for valuable discussions. 
R.O.\ thanks Norah Tan for helpful discussions.
C.W.\ thanks Ewin Tang for sharing an earlier version of \cite{pelecanos2026keyl}.
Part of this work was carried out while the authors visited the African Institute for Mathematical Sciences, Cape Town, for the 1st AIMS Workshop on the Theory of Quantum Learning Algorithms (2025).

\subsection*{AI Use Disclosure}

We can identify two central ideas of this work: a)~exploiting that several classic random matrix ensembles have closed-form expressions for the tensor moments in terms of Jucys--Murphy elements; b)~using the tilted law to construct ensembles of states whose average tensor moments are proportional to those of the random matrix ensembles. These were found and developed by the authors after several other failed attempts. During this phase, Gemini 3.1 Pro and ChatGPT 5.5 were used as a computational aid, and we produced examples similar to the one in the current warm-up section. Subsequently, interactions with ChatGPT 5.5 and 5.6 together with substantial human input led to extending the method to moment-matching of arbitrary degree, and the current manuscript is the result of the work of the authors to give the most elementary version of the arguments. The authors take full responsibility for the content of the manuscript.

\section{Preliminaries}\label{sec:prelim}

A $d$-dimensional quantum state $\rho$ is described by a density matrix, i.e., a positive semidefinite operator with unit trace. We denote the space of $d$-dimensional states as $\mathcal{D}_d$. The spectrum of a quantum state is the list of its eigenvalues, and is denoted by $\spec(\rho)$; note that this forms a probability distribution by virtue of $\rho$ being a density operator. The central task considered in this work is estimating this spectrum; let us first state the metric with respect to which we will characterize this task.

\begin{definition}
\label{def:sorted-total-variation}
    Given two probability distributions $p,q$, their total variation distance is given by $\dtv{p}{q} \coloneqq \frac12 \|p - q\|_1$. We will also consider the sorted-TV distance,  $\dtvsorted{p}{q} \coloneqq \dtv{p^\downarrow}{q^\downarrow}$, where $p^\downarrow,q^\downarrow$ are distributions formed by sorting the entries of $p,q$ in non-increasing order.
\end{definition}

We will now describe important distances between quantum states. First, let us define useful norms of matrices.

\begin{definition}
    For $p \geq 1$, we write $\|A\|_p$ for the Schatten $p$-norm of matrix~$A$.
\end{definition}

Let us now recall some standard distances between states:
\begin{definition}
    Let $\rho, \sigma$ be density matrices.  We write $\Dtr{\rho}{\sigma} = \frac12\|\rho - \sigma\|_1$ for their \emph{trace distance}.
\end{definition}

We will now define normalized power sums of distributions and states.

\begin{definition}
    For $L \in \N$ and a probability distribution $\alpha$ on $[d]$, we write \begin{equation}
        \mathrm{pow}_L(\alpha) = d^{L-1}  \cdot \sum_{i=1}^d \alpha_i^L \geq 1
    \end{equation}
    for the \emph{normalized $L$th power sum} of~$\alpha$.
    If $\rho$ is a density matrix, we also write $\mathrm{pow}_L(\rho)$ for $\mathrm{pow}_L(\spec(\rho))$.
\end{definition}

The utility of power sums is that they are a good proxy for sorted-TV distance:
\begin{proposition} \label{prop:power-sum-to-TV}
    Suppose $\alpha, \beta$ are probability distributions on~$[d]$ satisfying $\alpha_i, \beta_i \leq B/d$ for all~$i$.
    Then
    \begin{equation}
        \dtvsorted{\alpha}{\beta} \geq \frac{1}{2LB^{L-1}} \cdot \abs*{\mathrm{pow}_L(\alpha)-\mathrm{pow}_L(\beta)}.
    \end{equation}
\end{proposition}
\begin{proof}
    Without loss of generality, $\alpha, \beta$ are already sorted.  Then
    \begin{equation}
        \abs*{\mathrm{pow}_L(\alpha)-\mathrm{pow}_L(\beta)}
        \leq d^{L-1} \sum_{i=1}^d |\alpha_i^L - \beta_i^L|
        \leq d^{L-1}\sum_{i=1}^d L(B/d)^{L-1} |\alpha_i - \beta_i|
        = 2LB^{L-1}\cdot \dtvsorted{\alpha}{\beta},
    \end{equation}
    where the second inequality used that
    $x\mapsto x^L$ has derivative at most
    $L(B/d)^{L-1}$ on $[0,B/d]$.  
\end{proof}

Lastly, let us mention some notational conventions that will be followed throughout this paper. We use $\Tr$ to denote the trace of an operator. When the dimension $d$ is clear from context, we will use $\tr = \frac1d \Tr$ to denote the \emph{normalized} trace. We will use asymptotic notation $O,\Omega,o,\omega,\Theta$ with their
standard interpretation. Moreover, we will include variables in the subscript to suppress multiplicative factors depending only on such variables. For instance, $O_m(f(n,m)) = O(g(m) f(n,m))$ for some finite function $g$. We may also extend this notation naturally to multiple variables in the subscript or in the arguments.

\subsection{Permutations}
\label{sec:permutations}
We will denote by $S_n$ the symmetric group, i.e., the group of permutations over $n$ elements. Any permutation $\pi \in S_n$ can be decomposed into a set of cycles. For instance, the cycle notation $(1~2) (3)$ represents the permutation in $S_3$ that swaps the first two elements and leaves the third unchanged. 

\begin{definition}[Permutation notation]
    For any permutation $\pi \in S_n$, we define $\cyc(\pi)$ to be the set of its cycles and $\numcyc(\pi)$ to be their number. For a cycle $\tau \in \cyc(\pi)$, we use $\mathrm{len}(\tau)$ to denote its length. We define the cycle type of $\pi$ to be the sorted list of its cycle lengths. The Cayley length of a permutation is defined to be the minimum number of transpositions (2-cycles) necessary to implement $\pi$, and is denoted by $|\pi| \coloneqq n - \numcyc(\pi)$. The support of $\pi$ is defined as the set of elements it acts non-trivially on, i.e., $\supp(\pi) \coloneqq \{i \in [n]: \pi(i) \neq i\}$. 
\end{definition}

 We will often restrict our attention to the support of a permutation and omit the $1$-cycles; we define the non-fixed or non-trivial cycle type of a permutation as the cycle type with all $1$-cycles removed. E.g., a permutation with cycle type $(3,2,1,1)$ has non-fixed type $(3,2)$. Note that this omission of $1$-cycles does not affect the Cayley length; for instance, both cycle types $(3,2,1,1)$ and $(3,2)$ correspond to Cayley length $3$.

We will be concerned with the $(\mathbb{C}^{d})^{\otimes n}$-representations of $S_n$, and by a slight abuse of notation we will use the same notation for permutations and their representations as unitary operators on $(\mathbb{C}^{d})^{\otimes n}$. In particular, a permutation's action is given by
\begin{equation}
    \pi \ket{j_1, \dots, j_n} = \ket{j_{\pi^{-1}(1)}, \dots, j_{\pi^{-1}(n)}}, \quad \textnormal{for } j_1, \dots, j_n \in [d].
\end{equation}
We will also be interested in inner products of these permutation operators with other operators. Specifically, we will make use of the following standard fact:
\begin{fact}
\label{fact:perm-product}
    For operators $X_1, \dots, X_n \in \mathbb{C}^{d \times d}$ and a permutation $\pi \in S_n$, we have
    \begin{equation}
        \Tr(\pi \cdot X_1 \otimes \dots \otimes X_n) = \prod_{\tau \in \cyc(\pi)} \Tr\parens*{\prod_{i \in \tau} X_i}.
    \end{equation}
    In particular, for $X \in \mathbb{C}$, 
    \begin{equation} \label{eqn:tr}
        \Tr(\pi X^{\otimes n}) = \prod_{\tau \in \cyc(\pi)} \Tr(X^{\mathrm{len}(\tau)}).
    \end{equation}
\end{fact}

Lastly, we define the Jucys--Murphy elements of the symmetric group algebra, which will play a key role in our analysis. For $1 \leq t \leq n$, the $t$th Jucys--Murphy element is given by
\begin{equation}
    J_t = \sum_{1 \leq i < t} (i\ t), \quad \text{and we will find it convenient to write } \wt{J}_t = \tfrac1d J_t.
\end{equation}
Here $(i\ t)$ is often denoted $\swap_{i,t}$ in the quantum literature. We note here that all Jucys--Murphy elements are Hermitian and $d + J_t$ is invertible, which is part of~\Cref{prop:the-key}.

It is a standard fact that all Jucys--Murphy elements commute.
Moreover, we have the following generating function identity, the \emph{Jucys Identity}~\cite{JUCYS1974107}:
\begin{equation}    \label{eqn:jucys}
    (z + J_1)(z + J_2)\cdots(z+J_n) = \sum_{\pi \in S_n} z^{\numcyc(\pi)}\pi.
\end{equation}
This generalizes the more elementary identity
\begin{equation}
    z(z+1)\cdots (z+n-1) = \sum_{\pi \in S_n} z^{\numcyc(\pi)},
\end{equation}
and both can be given a simple proof by induction on~$n$.

We will also define power sums of the Jucys--Murphy elements; the $k$th Jucys--Murphy power sum is defined as
\begin{equation}
    \calP_k \coloneqq\sum_{i \in [n]} \wt{J}_i^k.
\end{equation}

\section{Random projections}
\label{sec:projections}
Consider a Haar-random rank-$r$ projection $\bPi$ on $\C^d$.  
This is the quantum analogue of an indicator random variable $1_{\bS}$ for a uniformly random subset $\bS \subseteq [d]$ with $|\bS| = r$.
Indeed, if we identify random variables on $[d]$ with diagonal matrices, and write $R = \smash[b]{\diag(\underbrace{1, \dots, 1}_{r \text{ times}}, 0, \dots 0)}$, then
\begin{equation}    \label{eqn:sim}
    1_{\bS} \equiv \bP R \bP^\dagger, \qquad 
    \bPi \equiv \bU R \bU^\dagger,
\end{equation}
where $\bP \sim S_d$ is a uniformly random permutation matrix and $\bU \sim U(d)$ is a Haar-random unitary.
Observe the following trivial computation:
\begin{fact}
    Let $j_1, \dots, j_n \in [d]$, and write $m$ for the number of distinct elements among them.  
    Then for $\bS \subseteq [d]$ uniformly random with $|\bS| = r$, the probability all of the $j_t$'s are in~$\bS$ is
    \begin{equation}
        \E[1_{\bS}^{\otimes n}(j_1, \dots, j_n)] = \frac{r(r-1)\cdots (r-m+1)}{d(d-1) \cdots (d-m+1)}.
    \end{equation}
\end{fact}
A crucial formula for us is the below noncommutative analogue of this fact. We could not find it explicitly in the literature, but we would judge it to be standard for experts.
\begin{proposition} \label{prop:the-key}
    For a Haar-random rank-$r$ projection $\bPi$ on $\C^d$,
    \begin{equation} \label{eqn:the-key}
        \E[\bPi^{\otimes n}] = \frac{(r+J_1)(r+J_2)\cdots(r+J_n)}{(d+J_1)(d +J_2)\cdots(d+J_n)}.
    \end{equation}
    Here, the ratio of operators makes sense since the Jucys--Murphy elements commute and each denominator factor $(d+J_t)$ is invertible.
    
\end{proposition}
\begin{proof}
    We use the notation $\bU$ and $R$ from \Cref{eqn:sim}.
    Collins and \'{S}niady~\cite[Prop.~2.3, (2)+(3)]{CS06} gave the unitary twirl formula
    \begin{equation}    \label{eqn:plugme}
        \E[\bU^{\otimes n} \cdot A \cdot  \bU^{\dagger \otimes n}] = \Phi(A) \Phi(\Id)^{-1}, \quad  \text{where }\Phi(A) \coloneqq \sum_{\pi \in S_n} \Tr(\pi^{-1} A) \pi.
    \end{equation}  
    Observe, using \Cref{eqn:tr} and $R^\ell = R$, that 
    \begin{equation}
        \Phi(R^{\otimes n}) = \sum_{\pi \in S_n} r^{\numcyc(\pi)} \pi = (r + J_1)(r+J_2) \cdots (r+J_n),
    \end{equation}
    where the last identity is Jucys's \Cref{eqn:jucys}. Replacing $r$ with $d$ in the above formula gives the value for~$\Phi(\Id)$. Plugging both of these into \Cref{eqn:plugme} gives the required formula for $\E[\bPi^{\otimes n}]$. In particular, each factor $(d+J_i)$ is invertible. 
\end{proof}
\begin{remark}
    We only needed \Cref{eqn:plugme} for $A$ of the form $X^{\otimes n}$, and in this case we can sketch an alternative proof.  Let $\bZ \in \C^{d \times d}$ be a standard complex Gaussian random matrix, with polar decomposition $\bZ = \bS^{1/2} \bU$, where $\bS = \bZ \bZ^\dagger$ and $\bU$ is an independent Haar-random unitary. Defining
    \begin{equation}
        \calE(X^{\otimes n}) \coloneqq \E[\bU^{\otimes n} \cdot X^{\otimes n} \cdot  \bU^{\dagger \otimes n}],
    \end{equation}
    we have
    \begin{equation} \label{eqn:combi1}
        \E[(\bZ X \bZ^\dagger)^{\otimes n}] = \E[(\bS^{1/2})^{\otimes n} \cdot \calE(X^{\otimes n}) \cdot (\bS^{1/2})^{\otimes n}] = \calE(X^{\otimes n})  \cdot \E[\bS^{\otimes n}],
    \end{equation}
    where the last equation used that the twirl $\calE(X^{\otimes n})$ commutes with any operator $Y^{\otimes n}$.  
    On the other hand, by Isserlis's/Wick's Theorem, it is not too hard to derive
    \begin{equation} \label{eqn:combi2}
        \E[(\bZ X \bZ^\dagger)^{\otimes n}] = \sum_{\pi \in S_n} \prod_{\tau \in \cyc(\pi)} \Tr(X^{\mathrm{len}(\tau)}) \pi = \Phi(X^{\otimes n}).
    \end{equation}
    Putting $X = \Id$ into this yields $\E[\bS^{\otimes n}] = \Phi(\Id)$, and the proof is complete by combination with \Cref{eqn:combi1,eqn:combi2}.
\end{remark}

\begin{definition} \label{def:a}
    Let $a = (a_1, \dots, a_K)$ be a sequence of positive integers and let $d$ be a multiple of each~$a_i$.  We write $\bX \sim \PRP_d(a)$ to denote that $\bX$ is drawn from the ``product of random projections'' distribution given by
    \begin{equation}
        \bX = \bW^\dagger \bW, \quad \bW \coloneqq \bPi_K \bPi_{K-1} \cdots \bPi_2 \bPi_1,
    \end{equation}
    where $\bPi_i$ is a Haar-random projection of rank $d/a_i$ on~$\C^{d}$, drawn independently for $i = 1 \dots K$\footnote{This assumes that each $a_i$ divides $d$. In general, one can embed these PRPs into a space with dimension $d^\prime \leq d$ for the largest possible $d^\prime$ divisible by each $a_i$; as the $a_i$s are dimension-independent constants, this would only affect our lower bounds by constant factors.}.
    It is easy to verify that
    \begin{equation}\label{eq:average1}
        \E[\bX] = \nu_a \cdot \Id, \qquad \nu_a \coloneqq 1/(a_1 a_2 \cdots a_K).
    \end{equation}
    (This is the $n = 1$ case of \Cref{prop:the-key}.)
\end{definition}

More generally, we can explicitly write the $n$-fold tensor moment.

\begin{proposition} \label{prop:f}
    With the notation of \Cref{def:a},
    for $\bX \sim \PRP_d(a_1, \dots, a_K)$, 
    \begin{equation}
        \E[\bX^{\otimes n}] = \nu_a^n \cdot \prod_{t=1}^n f_a(\wt{J}_t),
    \quad
    \text{where}
    \quad
        f_a(z) = \prod_{i=1}^K \frac{1+a_i z}{1+z}.
    \end{equation}
\end{proposition}
\begin{proof}
    We have
    \begin{equation}
        \E[\bX^{\otimes n}] = \E[\wh{\bW}^{\dagger \otimes n} \cdot \E[\bPi_K^{\otimes n}] \cdot \wh{\bW}^{\otimes n}], \quad \text{where }\wh{\bW} \coloneqq \bPi_{K-1} \cdots \bPi_1.
    \end{equation}
    By \Cref{prop:the-key}, $\displaystyle \E[\bPi_K^{\otimes n}] = \prod_{t=1}^n \frac{d/a_K + J_t}{d + J_t}$.  Moreover, since $\bPi_K$ is unitarily invariant, this commutes with $\wh{\bW}^{\otimes n}$ and $\wh{\bW}^{\dagger \otimes n}$.  Thus we can pull out this factor and continue by induction, concluding
    \begin{equation}
        \E[\bX^{\otimes n}] = \prod_{i=1}^K \prod_{t=1}^n \frac{d/a_i + J_t}{d + J_t}.
    \end{equation}
    Rearranging this completes the proof.
\end{proof}

\section{Warmup: rank-$d/2$ projectors vs Haar--random marginals}
\label{sec:warmup}

We first show a quite simple proof of a $\Omega(d^{4/3})$ lower bound for both spectrum and entropy estimation and for rank testing. We take $d$ to be even without loss of generality. The argument is based on the hardness of distinguishing between $n$ copies of a state sampled from either:

\begin{itemize}
\item[a)] Draw a Haar-random rank-$d/2$ projection
$\bPi$ on $\C^d$ and set
\begin{equation}
\brho_{\mathrm{proj}}
\coloneqq
\frac{2}{d}\bPi.
\end{equation}

\item[b)] 
Draw a Haar random unit vector $\ket{\bpsi}$ from $\C^d\otimes\C^d$ and set
\begin{equation}
\brho_{\mathrm{Haar}}
\coloneqq
\Tr_{2}\!\left[
\ketbra{\bpsi}{\bpsi}
\right],
\end{equation}
where $\Tr_2$ denotes the partial trace over the second tensor factor. 
\end{itemize}

For the first ensemble, the spectrum and entropy are fixed. The spectral and entropic properties of the second ensemble varying the dimensions of the subsystems have been studied in detail, first to model black hole evaporation~\cite{page1993average}. In~\cite{nechita2007asymptotics} several key properties were proved, including (Thm.~5) that the empirical distribution of rescaled eigenvalues of a sequence $\rho_d$ 
\begin{equation}
L_{d}(x)\coloneqq\frac{1}{d}\sum_{i=1}^{d}\delta_{{d\lambda_i}(\brho_{\mathrm{Haar}})}(x)
\end{equation}
converges almost surely in the topology of weak convergence to the Marchenko--Pastur density $\mu_{MP}$ as $d\rightarrow \infty$, where
\begin{equation}
\mu_{\mathrm{MP}}(dx)=\frac{1}{2\pi}\sqrt{\frac{4-x}{x}}\mathbb{I}_{[0,4]}(x)dx,
\end{equation}
where $\mathbb{I}_{[0,4]}$ is the indicator function on $[0,4]$.
The exact formulas for the mean and variance of the entropy of a
Haar-induced state imply, in the balanced case,
\begin{equation}
\E S(\brho_{\mathrm{Haar}})
=
\log d-\frac12+O(d^{-2}),
\qquad
\Var\parens*{S(\brho_{\mathrm{Haar}})}
=
\frac{1}{4d^2}+O(d^{-4}).
\label{eq:warmup-entropy-asymptotics}
\end{equation}
The mean formula was conjectured in~\cite{page1993average} and proved, for example,
by Sen~\cite{sen1996average}. The exact variance
formula was conjectured by Vivo, Pato, and
Oshanin~\cite{vivo2016random} and proved by
Wei~\cite{wei2017proof}; see also~\cite{bianchi2019typical}
for a later independent derivation and generalization.

We can then prove the following lemma.

\begin{lemma}
\label{lem:warmup-separation}
Let
\begin{equation}
{\balpha}_{d/2}
=
\Bigl(
\underbrace{\frac{2}{d},\ldots,\frac{2}{d}}_{d/2},
\underbrace{0,\ldots,0}_{d/2}
\Bigr).
\end{equation}
There exists a universal constant $\eps_0>0$ such that, for
$\brho_{\mathrm{Haar}}$, with probability tending
to one as $d\to\infty$ through even integers,
\begin{align}
\dtvsorted{
\spec(\brho_{\mathrm{Haar}})
}{
{\balpha}_{d/2}
}
&\geq \eps_0,
\\
S(\brho_{\mathrm{Haar}})
-
\log(d/2)
&\geq \eps_0,
\\
\inf_{\substack{\rho\in\mathcal D_d\\
\mathrm{rank}(\rho)\leq d/2}}
\Dtr{\brho_{\mathrm{Haar}}}{\rho}
&\geq \eps_0.
\end{align}
\end{lemma}
\begin{proof}
Set $\Delta_{\mathrm{ent}}=\log 2-\frac{1}{2}>0$.
The entropy separation is immediately proven by~\Cref{eq:warmup-entropy-asymptotics} and Chebyshev inequality:
\begin{equation}
\Pr\!\left[
S(\brho_{\mathrm{Haar}})
\leq
\log(d/2)+\frac{\Delta_{\mathrm{ent}}}{2}
\right]
=
O(d^{-2}).
\end{equation}
Via the convergence to the Marchenko--Pastur law one, for the interval $I=[\frac{1}{4},\infty)$ (which is valid as test function because $\mu_{MP}$ has no atoms at the endpoints),
\begin{equation}
\lim_{d\rightarrow \infty}L_{d}(I)=\mu_{MP}(I)=\beta>0\, ,\, \mathrm{almost\,surely.}
\end{equation}

In particular, one has the elementary bound
\begin{align}
1-\beta=\mu_{\mathrm{MP}}\!\left([0,1/4]\right)
&=
\int_0^{1/4}
\frac{1}{2\pi}
\sqrt{\frac{4-x}{x}}\,dx
\nonumber\leq
\int_0^{1/4}
\frac{1}{\pi\sqrt{x}}\,dx
=
\frac{1}{\pi}<\frac{1}{2}.
\end{align}
Since almost sure convergence implies convergence in probability, 
 fixed any $\delta\in(0,1)$ and $\gamma\in(1/2,\beta)$
 , there is a $d_0$ large enough such that for every $d\geq d_0$, with probability larger than $1-\delta$, the fraction of eigenvalues in $[\frac{1}{4d},\infty)$ of the sequence $\brho_{\mathrm{Haar}}$ is larger than $\gamma$. Since only $\frac{d}{2}$ eigenvalues can be in the larger half of the spectrum, this means that at least $(\gamma-\frac{1}{2})d$ are in the lower half and larger than $\frac{1}{4d}$. Therefore, their mass $\tau_d$ is at least $\frac{\gamma-\frac{1}{2}}{4}$. Consequently, for every
\begin{equation}
c<
\frac14\left(\beta-\frac12\right),
\end{equation}
we have
\begin{equation}
\Pr[\tau_d(\brho_{\mathrm{Haar}})\geq c]
\longrightarrow1.
\end{equation}
For any probability vector $q$ supported on at most $d/2$ coordinates, one has the elementary inequality
\begin{equation}
\left\|
\spec(\brho_{\mathrm{Haar}})^\downarrow-q
\right\|_1
\geq
2\tau_d.
\label{eq:spectral-tail-distance}
\end{equation}
Moreover, for every density matrix $\rho$ with decreasing eigenvalues $\lambda_1,\ldots,\lambda_d$,
\begin{equation}
\inf_{\substack{\sigma\in\mathcal D_d\\
\mathrm{rank}(\sigma)\leq d/2}}
\|\rho-\sigma\|_1
=
2\sum_{i>d/2}\lambda_i(\rho),
\label{eq:rank-approximation}
\end{equation}
where the lower bound follows by measuring the support projector of $\sigma$ and applying the Ky Fan maximum principle, while the upper bound is from choosing the normalized restriction of $\rho$ to its largest $\frac{d}{2}$ eigenspaces. Therefore, with $\tau_d= \sum_{i>d/2}\lambda_i(\brho_{\mathrm{Haar}})$, we have that $\brho_{\mathrm{Haar}}$ is also far from any state of rank at most $\frac{d}{2}$. The claim follows by a union bound. 
\end{proof}

Let the average state of the first ensemble be 
\begin{equation}
\ol{\rho}_{\mathrm{proj}}^{(n)}
\coloneqq
\E\!\left[\brho_{\mathrm{proj}}^{\otimes n}\right].
\end{equation}
Applying \Cref{prop:the-key} with $r=d/2$ gives
\begin{equation}
\ol{\rho}_{\mathrm{proj}}^{(n)}
=
\left(\frac{2}{d}\right)^n
\prod_{t=1}^n
\frac{1+2\wt{J}_t}{1+\wt{J}_t}.
\label{eq:warmup-projection-moment}
\end{equation}

The following lemma expresses the average state of the second ensemble in terms of Jucys--Murphy elements

\begin{lemma}\label{lem:pagestates}
The average $n$-fold state of $\brho_{\mathrm{Haar}}$ satisfies:

\begin{equation}
\ol{\rho}_{\mathrm{Haar}}^{(n)}
\coloneqq
\E\!\left[
\brho_{\mathrm{Haar}}^{\otimes n}
\right]=\frac{d^n}{(d^2)_n}\prod_{t=1}^n(1+\wt{J}_t).
\end{equation}
\end{lemma}

\begin{proof}
We identify
\begin{equation}
\left(\C^d\otimes\C^d\right)^{\otimes n}
\cong
\left(\C^d\right)^{\otimes n}
\otimes
\left(\C^d\right)^{\otimes n},
\end{equation}
and let $\Tr_2$ denote the partial trace over the second factor
in this decomposition. The Haar-random average of $\ketbra{\bpsi}{\bpsi}^{\otimes n}$ is proportional to the projector onto the symmetric subspace (of dimension $D_{n,d^2}=\frac{(d^2)_{n}}{n!}$, where $(x)_n
\coloneqq
x(x+1)\cdots(x+n-1)$):
\begin{equation}
\Ex \!\left[\ketbra{\bpsi}{\bpsi}^{\otimes n}\right]=\frac{1}{D_{n,d^2}}\frac{1}{n!}\sum_{\pi\in S_n}\pi\otimes \pi=\frac{1}{(d^2)_n}\sum_{\pi\in S_n}\pi\otimes \pi\,.
\end{equation}
By taking the partial trace,
\begin{equation}
\ol{\rho}_{\mathrm{Haar}}^{(n)}=\E\!\left[
\brho_{\mathrm{Haar}}^{\otimes n}
\right]=\Tr_{2}\left[\E \!\left[\ketbra{\bpsi}{\bpsi}^{\otimes n}\right]\right]=\frac{1}{(d^2)_n}\sum_{\pi}d^{\numcyc(\pi)}\pi=\frac{1}{(d^2)_n}\prod_{t=1}^n(d+J_t) \,.
\end{equation}
where in the last equality we used the Jucys identity~\Cref{eqn:jucys}.
\end{proof}

The following lemma gives relevant information on the spectrum of Jucys--Murphy elements. The claim can be easily obtained from the characterization of the spectrum of Jucys--Murphy as contents of a Young diagram, but we give a self-contained proof here.

\begin{lemma}
\label{lem:warmup-support-spectrum}
The states $\ol{\rho}_{\mathrm{proj}}^{(n)}$ and
$\ol{\rho}_{\mathrm{Haar}}^{(n)}$ commute. Moreover,
$\ol{\rho}_{\mathrm{Haar}}^{(n)}$ is invertible and, on the support
of $\ol{\rho}_{\mathrm{proj}}^{(n)}$,
\begin{equation}
J_t>-\frac d2
\qquad
\text{for every $t\in[n]$}.
\label{eq:warmup-Jucys--Murphy-support}
\end{equation}
\end{lemma}

\begin{proof}
By~\Cref{lem:pagestates} and~\Cref{prop:the-key}, $\ol{\rho}_{\mathrm{Haar}}^{(n)}$ is invertible. $\ol{\rho}_{\mathrm{proj}}^{(n)}$
and
$\ol{\rho}_{\mathrm{Haar}}^{(n)}$ commute because they are functions of Jucys--Murphy elements.
For $q>0$ and $t\in[n]$, define
\begin{equation}
\mathcal G_{q,0}\coloneqq\Id,
\qquad
\mathcal G_{q,t}\coloneqq\prod_{j=1}^t(q+J_j).
\end{equation}

By the random-projection moment formula in
\Cref{prop:the-key},
\begin{equation}
\E\left[\bPi^{\otimes t}\right]
=
\mathcal G_{d/2,t} [(d^2)_t\ol{\rho}_{\mathrm{Haar}}^{(t)}]^{-1}.
\label{eq:warmup-projection-gram}
\end{equation}
All the operators in~\eqref{eq:warmup-projection-gram} commute and
\begin{equation}
\mathcal G_{d/2,t}
=
\E\left[\bPi^{\otimes t}\right]
[(d^2)_t\ol{\rho}_{\mathrm{Haar}}^{(t)}]
\succeq0.
\label{eq:warmup-gram-positive}
\end{equation}
Here we used that both factors on the right-hand side are
positive semidefinite and commute. Moreover, from the product
representation of $\mathcal G_{d/2,t}$, the support of $\mathcal G_{d/2,n}$ is included in the support of each $\mathcal G_{d/2,t}$ for $t\leq n$.
Using
\begin{equation}
\mathcal G_{d/2,t}
=
\mathcal G_{d/2,t-1}
\parens*{\frac d2+J_t}
\end{equation}
and commutativity, we obtain that on the support of $\mathcal G_{d/2,n}$
\begin{equation}
\frac d2+J_t
=
\mathcal G_{d/2,t-1}^{-1}
\mathcal G_{d/2,t}
\succ0,
\end{equation}
where the inverse is taken on the support of $\mathcal G_{d/2,n}$.
Since $\ol{\rho}_{\mathrm{proj}}^{(n)}
=
\left(\frac{2}{d}\right)^n
\mathcal G_{d/2,n}\mathcal G_{d,n}^{-1}$, we have that $J_t>-d/2$ on the support of
$\ol{\rho}_{\mathrm{proj}}^{(n)}$.
\end{proof}

We now prove the following:

\begin{theorem}\label{thm:warmup-relative-entropy}
There exists a constant $C>0$ such that for any even $d$ and $1\leq n\leq d^2$,
\begin{equation}
D\!\left(
\ol{\rho}_{\mathrm{proj}}^{(n)}
\middle\|
\ol{\rho}_{\mathrm{Haar}}^{(n)}
\right)
\leq
C\frac{n^3}{d^4}\,.
\end{equation}
\end{theorem}
\begin{proof}
Since $\ol{\rho}_{\mathrm{Haar}}^{(n)}$ is invertible, we can restrict the evaluation of the relative entropy to the support of $\ol{\rho}_{\mathrm{proj}}^{(n)}$:
\begin{align}
\label{ineq:kl0}
D\!\left(
\ol{\rho}_{\mathrm{proj}}^{(n)}
\middle\|
\ol{\rho}_{\mathrm{Haar}}^{(n)}
\right)
&=
\Tr\left[
\ol{\rho}_{\mathrm{proj}}^{(n)}
\left(
\log\ol{\rho}_{\mathrm{proj}}^{(n)}
-
\log\ol{\rho}_{\mathrm{Haar}}^{(n)}
\right)
\right]
\nonumber\\
&=
\Tr\left[
\ol{\rho}_{\mathrm{proj}}^{(n)}
\log\left(
\frac{(d^2)_n}{d^{2n}}
\prod_{t=1}^n
\frac{1+2\wt{J}_t}
{\left(1+\wt{J}_t\right)^2}
\right)
\right].
\end{align}

On the support of $\ol{\rho}_{\mathrm{proj}}^{(n)}$, the minimum eigenvalues of $J_t$ are larger than $-d/2$ by~\Cref{lem:warmup-support-spectrum}. 
Thus, the minimum eigenvalue of $\wt{J}_t$ is strictly larger than $-1/2$. We now prove that
\begin{equation}
\log\frac{1+2s}{(1+s)^2}\leq -s^2+2s^3\,, \quad \forall s\in(-1/2,\infty)
\end{equation}
Indeed, let $\mathcal{L}(s)\coloneqq  -s^2+2s^3-\log\frac{1+2s}{(1+s)^2}$.
We have $\mathcal{L}(0)=0$ and $\frac{d}{ds}\mathcal{L}=\frac{2s^3(7+6s)}{(1+s)(1+2s)}$, which is positive for $s\geq 0$ and negative for $s\in(-1/2,0]$. Therefore, $\mathcal{L}(s)$ is minimised at $s=0$ in $(-1/2,\infty)$. This implies

\begin{align}\label{ineq:boundkl}
D\!\left(
\ol{\rho}_{\mathrm{proj}}^{(n)}
\middle\|
\ol{\rho}_{\mathrm{Haar}}^{(n)}
\right)
\leq \log\frac{(d^2)_n}{d^{2n}}-\sum_{t=1}^n\Tr\left[\ol{\rho}_{\mathrm{proj}}^{(n)}\wt{J}_t^2\right]+2\sum_{t=1}^n\Tr\left[\ol{\rho}_{\mathrm{proj}}^{(n)}\wt{J}_t^3\right] \,.
\end{align}
The first term can be bounded as
\begin{equation}\label{ineq:coeff}
\log\frac{(d^2)_n}{d^{2n}}=\sum_{t=0}^{n-1}\log\left(1+\frac{t}{d^2}\right)\leq \frac{n(n-1)}{2d^2}\,.
\end{equation}
We now have to evaluate the expectation values of $\mathcal{P}_k$ for $k=2$ and $k=3$. From the expressions of the Jucys--Murphy elements, we have

\begin{itemize}
\item $\mathcal{P}_2= \sum_{t=2}^n\sum_{i,j<t} (i\ t) (j\ t)$, and the terms with $i=j$ are identities, while the terms with $i\neq j$ are 3-cycles. Therefore
\begin{equation}\label{ineq:j2}
\sum_{t=1}^n\Tr\left[\ol{\rho}_{\mathrm{proj}}^{(n)}\wt{J}_t^2\right]=\frac{n(n-1)}{2d^2}+\frac{n^3-3n^2+2n}{3d^2} \Ex\Tr[\brho_{\mathrm{proj}}^3]=\frac{n(n-1)}{2d^2}
+
\frac{4n^3-12n^2+8n}{3d^4}\,.
\end{equation}
\item $\mathcal{P}_3= \sum_{t=2}^n\sum_{i,j,k<t} (i\ t) (j\ t) (k\ t)$, and the terms with $i,j,k$ all different are 4-cycles, while the others are swaps. Therefore
\begin{align}\label{ineq:j3}
\sum_{t=1}^n\Tr\left[\ol{\rho}_{\mathrm{proj}}^{(n)}\wt{J}_t^3\right]&=\frac{(n-3)(n-2)(n-1)n}{4d^3}\Ex\Tr[\brho_{\mathrm{proj}}^4]+\frac{2n^3-5n^2+3n}{2d^3}\Ex\Tr[\brho_{\mathrm{proj}}^2]\\
&=\frac{2(n-3)(n-2)(n-1)n}{d^6}+\frac{2n^3-5n^2+3n}{d^4}\,.
\end{align}
\end{itemize}
Combining~\Cref{ineq:coeff}, ~\Cref{ineq:j2},~\Cref{ineq:j3}, with~\Cref{ineq:boundkl}, the $\frac{n(n-1)}{2d^2}$ terms cancel and ignoring negative terms, we obtain
\begin{equation}
D(\ol{\rho}_{\mathrm{proj}}^{(n)}\|\ol{\rho}_{\mathrm{Haar}}^{(n)})\leq 2 \sum_{t=1}^n\Tr\left[\ol{\rho}_{\mathrm{proj}}^{(n)}\wt{J}_t^3\right]=\frac{4(n-3)(n-2)(n-1)n}{d^6}+\frac{4n^3-10n^2+6n}{d^4}\,.
\end{equation}
which, by inspection, proves the claim.
\end{proof}

Pinsker's inequality and
\Cref{thm:warmup-relative-entropy} show that
\begin{equation}
\Dtr{
\ol{\rho}_{\mathrm{proj}}^{(n)}
}{
\ol{\rho}_{\mathrm{Haar}}^{(n)}
}
=
o(1)
\qquad
\text{whenever $n=o(d^{4/3})$}.
\label{eq:warmup-indistinguishability}
\end{equation}
Together with \Cref{lem:warmup-separation}, and the standard reduction to state discrimination, this proves the following
warm-up consequence.

\begin{corollary}
\label{cor:warmup-lower-bound}
There is a universal constant $\eps>0$ such that estimating the
spectrum or the entropy to precision $\eps$ requires
$\Omega(d^{4/3})$ copies.  The same pair of ensembles gives an
$\Omega(d^{4/3})$ lower bound for testing whether a state has rank at
most $d/2$ or is at trace distance at least $\eps$ from every such
state.
\end{corollary}

\section{Construction and Proofs of Main Results}
\label{sec:contruction}

We now present the pair of hard instances used to prove our main lower bounds. Our construction will involve products of random projections as defined in \Cref{sec:projections}. In particular, we will draw operators from $\PRP_d(a)$ and $\PRP_d(b)$, for two different sequences $a,b \in \mathbb{Z_+^K}$. To prove statistical indistinguishability of these state mixtures, we will analyze their log-likelihood ratios. Moreover, we will design our mixtures such that $\Ex [\brho^{\otimes n}] \propto \Ex[\bX^{\ot n}]$, where $\bX$ is drawn from $\PRP_d(a)$ or $\PRP_d(b)$. Thus, by \Cref{prop:f}, we are motivated to consider

\begin{equation} \label{eqn:h}
    h(z) \coloneqq \log\frac{f_b(z)}{f_a(z)} = \sum_{j=1}^\infty \tfrac{(-1)^{j+1}}{j}\parens*{\sum_{i=1}^K b_i^j - \sum_{i=1}^K a_i^j}z^j.
\end{equation}

Indistinguishability will arise when the Taylor series coefficients vanish until some degree~$k$.  With our argument, we get an $\Omega(d^{4/3})$ lower bound for $k = 2$, 
no improvement for $k = 3$, an $\Omega(d^{3/2})$ lower bound for $k = 4$, etc.; in general, the exponent on~$d$ will be $2-2/(\lfloor k/2 \rfloor + 2)$.  Thus we will only consider even~$k$, for notational simplicity.

Now, for even $k$, we wish to find sequences $a,b \in \Z_+^K$ (for some $K = K(k)$) that have matching power-sums until degree~$k$.  
\begin{example}
    For $k = 2$ we may take $K = 2$ and $a = (1,3)$, $b = (2,2)$; here $1^1+3^1 = 2^1+2^1$ but $1^2+3^2 \neq 2^2 + 2^2$.  Notice that with entries of $1$ like $a_1 = 1$, the associated random projection is just the identity operator. So for this example, we are comparing a random rank-$d/3$ projection with the sandwiched-product of two random rank-$d/2$ projections.
\end{example}
Finding integer sequences with matching power-sums is the well-known Prouhet--Tarry--Escott problem, and a simple construction based on the Thue--Morse sequence is known which achieves $K = 2^{k-1}$:
\begin{proposition} \label{prop:prouhet}
    For $k \in \N$ and $K = 2^{k-1}$, there exist sequences $a,b \in \Z_+^K$ with $\sum_{i=1}^K a_i^j = \sum_{i=1}^K b_i^j$ for all $j < k$ and $\sum_{i=1}^K a_i^k \neq \sum_{i=1}^K b_i^k$.
\end{proposition}
\begin{proof}
    For $r \in \{0, \dots, 2^k-1\}$, let $s_2(r)$ denote the sum of the binary digits of $r$. The Thue--Morse sequence partitions $\{0, \dots, 2^k-1\}$ into two sets $O$ and $E$ based on the parity of $s_2(r)$, each with $K$ elements. The differences of power sums of these sets are generated by the function
    \begin{equation}
        F(z) \coloneqq \sum_{r = 0}^{2^k - 1} (-1)^{s_2(r)} e^{rz}.
    \end{equation}
    The $j$th derivative of $F$ satisfies $F^{(j)}(0) = \sum_{r \in E} r^j - \sum_{r \in O} r^j$. Moreover, by considering the bitwise expansion of each $r$, one can show that $F$ can be rewritten as
    \begin{equation}
    \label{eq:prop-prouhet-1}
        F(z) = \prod_{j = 0}^{k-1} (1-e^{2^j z}).
    \end{equation}
    Writing $e^x = 1 + x + \dots$, it is apparent that the expression on the RHS of \Cref{eq:prop-prouhet-1} has no terms below degree $k$, but does have a non-zero degree-$k$ term. Consequently, for $j < k$, $\sum_{r \in E} r^j - \sum_{r \in O} r^j = 0$, and these moments are separated at degree $k$, as desired.

    While the sets $O$ and $E$ satisfy the desired moment-matching guarantees, $E$ contains $0$. To ensure all integers are positive, we simply increase all of them by $1$. One can verify that this preserves the moment-matching guarantees using, say, the binomial theorem.

\end{proof}

Combining this with \Cref{eqn:h} yields:
\begin{proposition} \label{prop:not}
    For $k$, $K$, $a$, $b$ as in \Cref{prop:prouhet}, there are $\beta_k \neq 0$, $\beta_{k+1}$ and $c_{k}>0$ such that
    \begin{equation}
         \left|\log\frac{f_b(z)}{f_a(z)}- \beta_k z^k - \beta_{k+1} z^{k+1}\right| \leq c_{k}z^{k+2}, \qquad \textnormal{whenever } |z| < \frac12\min_i\{a_i^{-1},b_i^{-1}\}.
    \end{equation}
\end{proposition}

Let us finally define the mixtures of states we will consider. A natural strategy would be to draw $\bX \sim \PRP_d$, and then normalize by $\Tr(\bX)$. However, understanding the $n$-fold tensor moments becomes much harder after such normalization. Moreover, to make use of \Cref{prop:not}, we wish to design mixtures such that these moments of the random states are proportional to those of the underlying products of random projections. We will thus draw $\bX$ from a ``tilted'' distribution designed to ensure this property.

\begin{definition}  \label{def:normalize}
    Fix $n \in \N$ and suppose $\bX \sim \PRP_d(a)$.  We write $\wt{\bX} \sim \wt{\PRP}^{(n)}_d(a)$ to denote that $\wt{\bX}$ is drawn from the tilted distribution whose density with respect to~$\bX$ is $\Tr(\bX)^n/\E[\Tr(\bX)^n]$.
    The random variable for a state sampled from the tilted law is denoted by $\brho\coloneqq\wt{\bX}/\Tr(\wt{\bX})$.
    Finally, the tilt is designed so that
    \begin{equation}
        \ol{\rho}^{(n)} \coloneqq \E[\brho^{\otimes n}] \propto \E[\bX^{\otimes n}] \propto \prod_{t=1}^n f_a(\wt{J}_t)
    \end{equation}
    (the last step using \Cref{prop:f}).
\end{definition}

The probability of an event under the untilted and tilted law are related as follows: for any measurable set $A$ with indicator function $\mathbb{I}_{A}$,

\begin{equation}\label{eq:tiltedprob}
\Pr(\wt{\bX}_e\in A)=\frac{\E[\Tr(\bX_{e})^n\mathbb{I}_{A}]}{\E[\Tr(\bX_{e})^n]}.
\end{equation}

A simple consequence of this expression is the following relation:

\begin{proposition}\label{prop:prob_tilt}
For any measurable set $A$,
\begin{equation}
\Pr(\wt{\bX}_e\in A)\leq \nu_e^{-n}\Pr({\bX}_e\in A)\,.
\end{equation}
\begin{proof}
Since $0\preceq\bX_e \preceq \Id$, $\Tr(\bX_e)^n\leq d^n$, while by Jensen's inequality and \Cref{eq:average1}, $\E [\Tr(\bX_e)^n]\geq \E [\Tr(\bX_e)]^n =(d\nu_e)^{n}$. Using these relations to bound the rhs in \Cref{eq:tiltedprob} proves the claim.
\end{proof}
\end{proposition}

\begin{notation}
    Henceforth in this paper we fix the notation from \Cref{prop:not}.  We also let $\bX_a \sim \PRP_d(a)$, $\bX_b \sim \PRP_d(b)$ and use the natural associated notation $\wt{\bX}_a$, $\wt{\bX}_b$, $\brho_a$, $\brho_b$, $\ol{\rho}^{(n)}_a$, $\ol{\rho}^{(n)}_b$ as in \Cref{def:normalize}.
\end{notation}
Our two hard-to-distinguish states will be $\brho_a$, $\brho_b$.  Our goals will be to show that, on one hand,
\begin{equation}
    n = o(d^{2-2/(k/2+2)}) \quad\implies\quad \Dtr{\ol{\rho}^{(n)}_a}{\ol{\rho}^{(n)}_b} = o(1);
\end{equation}
and, on the other hand, there exists $\eps_k > 0$ such that
\begin{equation}
    n = o(d^2) \quad\implies\quad \dtvsorted{\mathrm{spec}(\brho_a)}{\mathrm{spec}(\brho_b)} \geq \eps_k \text{ whp}.
\end{equation}

For our entropy estimation and rank-testing lower bounds, we will similarly show that the entropies and ranks of these states are separated with high probability. 

\subsection{Proofs of main results}

To prove our main results, we will first show the following statistical indistinguishability result.

\begin{theorem}
\label{thm:prp-statistical-indistinguishability}
    For $n = o(d^{2 - \frac{4}{k+4}})$,
    \begin{equation}
        \Dtr{\ol{\rho}_a^{(n)}}{\ol{\rho}_b^{(n)}} = o(1).
    \end{equation}
\end{theorem}

The proof of the above theorem is deferred to \Cref{sec:indistinguishability}. We will use this to prove both \Cref{thm:intro-spec,thm:intro-entropy}. First, we will need the following spectral separation guarantee.

\begin{proposition}
\label{prop:typical-spectral-separation}
Assume that $n=o(d^2)$. There exist sets
$\mathcal T_a^{(d)}$ and $\mathcal T_b^{(d)}$ of spectra and a
constant $\delta_k^{\mathrm{spec}}>0$ such that, for every
$e\in\{a,b\}$,
\begin{equation}
\Pr\left(
\spec(\brho_e)\in\mathcal T_e^{(d)}
\right)
\geq
1-\exp\left(-\Omega_k(d^2)\right),
\label{eq:typical-spectrum-probability}
\end{equation}
and
\begin{equation}
\dtvsorted{\alpha}{\beta}
\geq
\delta_k^{\mathrm{spec}}
\label{eq:typical-spectrum-distance}
\end{equation}
for every $\alpha\in\mathcal T_a^{(d)}$ and
$\beta\in\mathcal T_b^{(d)}$. In particular,
\begin{equation}
\Pr\left(
\dtvsorted{
\spec(\brho_a)
}{
\spec(\brho_b)
}
\geq
\delta_k^{\mathrm{spec}}
\right)
\geq
1-\exp\left(-\Omega_k(d^2)\right).
\label{eq:random-spectrum-separation}
\end{equation}
\end{proposition}

The proof of this spectral separation is deferred to \Cref{sec:spectral-separation}. Given \Cref{thm:prp-statistical-indistinguishability,prop:typical-spectral-separation}, our spectrum estimation lower bound follows immediately.

\begin{proof}[Proof of \Cref{thm:intro-spec}]
    Suppose we have an algorithm for spectrum estimation to within precision $\delta_k^{\spec}/3$ (defined as in \Cref{prop:typical-spectral-separation}), which succeeds with probability at least $\frac23$. By \Cref{prop:typical-spectral-separation}, this can be used to distinguish between $\brho_a$ and $\brho_b$ by simply choosing the closer of the two sets $\calT_a^{(d)}$ and $\calT_b^{(d)}$. Such a test fails only if either the spectrum estimation algorithm fails or if $\spec(\brho_e) \notin \calT_e^{(d)}$. Thus, by a union bound, the total failure probability is at most $\frac13 + \exp(-\Omega_k(d^2)) \leq \frac25$ for $d$ sufficiently large. Consequently, this tester can distinguish between $\brho_a$ and $\brho_b$ with probability at least $\frac35$; by \Cref{thm:prp-statistical-indistinguishability}, this must have the claimed copy complexity. 
    
\end{proof}

Next, we will show that states from these ensembles also have separated entropies.

\begin{proposition}
\label{prop:typical-entropy-separation}

Let $N_p$ be the depolarizing channel, i.e., $N_p(\rho) \triangleq p\rho + (1-p)\frac{\mathbbm{1}}{d}$. Assume that $n=o(d^2)$. There exists a constant
$p_k\in(0,1)$, sets $\mathcal U_a^{(d)}$ and
$\mathcal U_b^{(d)}$ of states, and a constant
$\delta_k^{\mathrm{ent}}>0$ such that, for every $e\in\{a,b\}$,
\begin{equation}
\Pr\left(
\brho_e\in\mathcal U_e^{(d)}
\right)
\geq
1-\exp\left(-\Omega_k(d^2)\right),
\label{eq:typical-entropy-probability}
\end{equation}
and
\begin{equation}
\left|
S\left(\mathcal{N}_{p_k}(\rho_a)\right)
-
S\left(\mathcal{N}_{p_k}(\rho_b)\right)
\right|
\geq
\delta_k^{\mathrm{ent}}
\label{eq:typical-entropy-separation}
\end{equation}
for every $\rho_a\in\mathcal U_a^{(d)}$ and
$\rho_b\in\mathcal U_b^{(d)}$. In particular,
\begin{equation}
\Pr\left(
\left|
S\left(\mathcal{N}_{p_k}(\brho_a)\right)
-
S\left(\mathcal{N}_{p_k}(\brho_b)\right)
\right|
\geq
\delta_k^{\mathrm{ent}}
\right)
\geq
1-\exp\left(-\Omega_k(d^2)\right).
\label{eq:random-entropy-separation}
\end{equation}
\end{proposition}

The proof of this entropy separation is deferred to \Cref{sec:entropy-separation}. Note that the separation does not hold for the original ensembles, but only after the states undergo a depolarizing channel; however, by the data-processing inequality, these states can only be more indistinguishable. Thus, one can obtain a proof of  \Cref{thm:intro-entropy} that is nearly identical to that of \Cref{thm:intro-spec}, except with \Cref{prop:typical-entropy-separation} instead of \Cref{prop:typical-spectral-separation}, which we omit.

Lastly, we show the desired rank separation.

\begin{proposition}[Rank separation]
\label{prop:rank-separation}
Assume $n=o(d^2)$ and set $r = d/2^k$. There is a constant $\delta^{\mathrm{rank}}>0$ such that
\begin{equation}
\label{eq:yes-rank}
\operatorname{rank}(\brho_a)=r
\qquad\text{almost surely},
\end{equation}
and except with probability at most $\exp(-\Omega(d^2))$, it holds that $\Dtr{\brho_b}{\sigma} \geq \delta^{\mathrm{rank}}$ for all $d$-dimensional states $\sigma$ of rank at most~$r$.
\end{proposition}
The proof of this is deferred to \Cref{sec:rank-separation}.  For the specific value $\beta = 1/2^k$, \Cref{thm:intro-rank} immediately follows, as with \Cref{thm:intro-spec}.  One can then obtain any other constant $0 < \beta < 1$ by a suitable simple padding of the dimension~$d$.

\section{Statistical Indistinguishability}
\label{sec:indistinguishability}

The main result of this section is \Cref{thm:prp-statistical-indistinguishability}, which shows that our mixtures are statistically indistinguishable. To prove indistinguishability, we will appeal to the triangular discrimination.

\begin{definition}[Triangular Discrimination]
    \label{def:triangular-discrimination}
    Given two discrete distributions $p = (p_y),q = (q_y)$ on some finite domain $\Omega$, define
    \begin{equation}
        \mu_y \coloneqq \frac{p_y + q_y}{2}, \quad \xi_y \coloneqq \frac{q_y - p_y}{q_y + p_y},
    \end{equation}
    setting $\xi_y = 0$ when $p_y = q_y = 0$. Then, the triangular discrimination between $p$ and $q$ is given by 
    \begin{equation}
        \Ex_{y \sim \mu} \xi_y^2 = \frac{1}{2}\sum_y \frac{(q_y-p_y)^2}{p_y + q_y},
    \end{equation}
    and $\xi_y$ satisfies
    \begin{equation}
    \label{eq:dtv-triangular-discrimination}
        \Ex_{y \sim \mu} \xi_y = 0, \quad \Ex_{y \sim \mu} |\xi_y| = \dtv{p}{q}.
    \end{equation} 
\end{definition}

The triangular discrimination is an $f$-divergence, and is equivalent up to constant factors to the Hellinger squared distance as well as the Jensen-Shannon $\chi^2$-divergence. We will relate this triangular discrimination to the log-likelihood ratio using the following elementary lemma.

\begin{lemma}
\label{lem:triangular-upper}
    Let $G \subseteq \Omega$ be a set on which $p_y,q_y > 0$. Fix some $\lambda \in \mathbb{R}$, and define
    \begin{equation}
       \mathrm{bad}(G) \coloneqq \mu(G^c), \quad\mathrm{dev}_{\lambda,G} \coloneqq \Ex_\mu \left[\left(\log \frac{q_y}{p_y} - \lambda\right)^2 \mathbbm{1}_G\right].
    \end{equation}
    Then,
    \begin{equation}
        \Ex_\mu \xi^2 \leq C(\mathrm{bad}(G) +\mathrm{dev}_{\lambda,G}),
    \end{equation}
    for some absolute constant $C > 0$.
\end{lemma}

\begin{proof}
    Define $\phi(x) \coloneqq \frac{e^x-1}{e^x+1}$ and $a \coloneqq \phi(\lambda)$. One can rewrite $\xi_y = \phi(\log(q_y/p_y))$. Note that $\phi^\prime(x) = \frac{2e^x}{(e^x+1)^2} \leq \frac12$. Consequently,
    \begin{equation}
        E \coloneqq \Ex_\mu [(\xi - a)^2 \mathbbm{1}_G] = \Ex_\mu [(\phi(\log(q_y/p_y)) - \phi(\lambda))^2 \mathbbm{1}_G] \leq \frac14 \Ex_\mu [(\log(q_y/p_y) -\lambda)^2 \mathbbm{1}_G] = \frac14\mathrm{dev}_{\lambda,G}.
    \end{equation}
    Let $W =  \mu(G) = 1 -\mathrm{bad}(G).$ Recall that $\Ex_\mu \xi = 0$. Thus,
    \begin{align}
        |a|W &= |\Ex_\mu[a \mathbbm{1}_G]| = |\Ex_\mu[a\mathbbm{1}_G - \xi]| = |\Ex_\mu[(a-\xi)\mathbbm{1}_G] - \Ex_\mu[\xi \mathbbm{1}_{G^c}]| \\&\leq |\Ex_\mu[(a-\xi)\mathbbm{1}_G]| + |\Ex_\mu[\mathbbm{1}_{G^c}]| \leq \sqrt{WE} +\mathrm{bad}(G),
    \end{align}
    where we used the triangle inequality and then Cauchy--Schwarz.

    Now, if we had $\mathrm{bad}(G) \geq \frac12$, the claimed bound would be immediate for $C = 2$ as $|\xi| \leq 1$. So, assume $\mathrm{bad}(G) \leq \frac12$, i.e., $W \geq \frac12$. Thus,
    \begin{equation}
        a^2 W \leq 2a^2W^2 \leq 4 W E + 4\mathrm{bad}(G)^2 \leq\mathrm{dev}_{\lambda,G} + 4\mathrm{bad}(G),
    \end{equation}
    as $W,\mathrm{bad}(G) \leq 1$ and $E \leq \frac14\mathrm{dev}_{\lambda,G}$. Finally, we have
    \begin{align}
        \Ex_\mu [\xi^2] = \Ex_\mu [\xi^2 \mathbbm{1}_G] + \Ex_\mu [\xi^2 \mathbbm{1}_{G^c}] &\leq 2 \Ex_\mu[(\xi - a)^2 \mathbbm{1}_G] + 2 \Ex[a^2 \mathbbm{1}_G] +\mathrm{bad}(G) \\&= 2E + 2a^2W +\mathrm{bad}(G) \\&\leq \frac52\mathrm{dev}_{\lambda,G} + 9\,\mathrm{bad}(G),
    \end{align}
    as desired.
\end{proof}

The above lemma allows us to split up the triangular discrimination into two components: 1) the deviation of the log-likelihood from a fixed scalar on the set where this deviation is sufficiently small, 2) a worst-case contribution when the deviation is large, but this occurs with small probability. Now, to handle the former case, \Cref{prop:f,prop:not} motivate us to bound high-order power sums of Jucys--Murphy elements. 

For the $j$th power-sum, let $\calP_j^{\circ} \coloneqq \calP_j - \eta_j \mathbbm{1}$, where $\eta_j$ is the component of $\calP_j$ associated with the identity permutation, i.e., we have removed all identity terms from $\calP_j$ to obtain $\calP_j^{\circ}$. Then, we show the following upper bound:

\begin{lemma}
    \label{lem:Jucys--Murphy-small-moments}
    For fixed $m \geq 2$, set
    \begin{equation}
        \vartheta_m \coloneqq 2 - \frac{2}{\lfloor m/2 \rfloor + 2}.
    \end{equation}
    Then, for $e \in \{a,b\}$,
    \begin{equation}
        \Tr(\ol{\rho}_e^{(n)} (\calP_m^\circ)^2) = o(1),
    \end{equation}
    whenever $n = o(d^{\vartheta_m})$.
\end{lemma}

Moreover, for the highest-order term appearing in \Cref{prop:not}, we will make use of the following bound for large even moments.

\begin{lemma}
    \label{lem:Jucys--Murphy-large-moments}
    For fixed $t \geq 1$ and $e \in \{a,b\}$,
    \begin{equation}
        \Tr(\ol{\rho}_e^{(n)} \calP_{2t}) \leq C_{k,t} \frac{n^{t+1}}{d^{2t}},
    \end{equation}
    whenever $n \leq c_k d^2$ for some constant $c_k > 0$.
\end{lemma}

Let us now use these lemmas to prove the main result of this section.
\begin{proof}[Proof of \Cref{thm:prp-statistical-indistinguishability}]
   Throughout this proof, we denote $\gamma \coloneqq \frac{4}{k+4}$.
   Now, to bound the trace distance between $\ol{\rho}_a^{(n)}$ and $\ol{\rho}_b^{(n)}$, first recall that these states can be expressed as operator-valued functions of Jucys--Murphy elements, and thus commute. Consequently, their trace distance is the same as the TV distance between their eigenvalues, and we can directly bound the latter. 

   Concretely, let $\Omega$ be some finite set that indexes the common eigenbasis of all Jucys--Murphy elements. For $y \in \Omega$ and $t \in [n]$, we define
   \begin{equation}
       \wt{j}_t(y) \coloneqq \bra{y} \wt{J}_t \ket{y}
   \end{equation}
   to be they $y$-th eigenvalue of $\wt{J}_t$. $\ol{\rho}_a^{(n)}, \ol{\rho}_b^{(n)}$ are also diagonal in this basis, and we define $\{p_y\}_{y \in \Omega}, \{q_y\}_{y \in \Omega}$ to be their respective eigenvalues. Defining $\mu, \xi$ as in \Cref{def:triangular-discrimination}, we can thus write
   \begin{equation}
       \Dtr{\ol{\rho}_a^{(n)}}{\ol{\rho}_b^{(n)}}^2 = \dtv{p}{q}^2 \leq \Ex_{y \sim \mu} \xi_y^2, \label{eq:dtv-triangular-2}
   \end{equation}
   where the inequality follows from \Cref{eq:dtv-triangular-discrimination} and Cauchy--Schwarz. It thus suffices to bound this triangular discrimination, for which we will use \Cref{lem:triangular-upper}.

   First, we will pick the set $G \subseteq \Omega$. We start by fixing some $\delta > 0$ satisfying
   \begin{equation}
       \delta < \frac12 \min_i\{a_i^{-1},b_i^{-1}\}.
   \end{equation}
   Then, the good set consists of all indices with bounded Jucys--Murphy eigenvalues:
   \begin{equation}
       G \coloneqq\{y \in \Omega: \max_{1 \leq t \leq n} |\wt{j}_t(y)| \leq \delta \}.
   \end{equation}
   Note that for each $y \in G$, $i \in [K]$, and $t \in [n]$,
   \begin{equation}
       |a_i \wt{j}_t(y)|, \textnormal{ } |b_i \wt{j}_t(y)| < \frac12 \implies p_y, q_y > 0,
   \end{equation}
   by \Cref{prop:f}. Define $B \coloneqq G^c$ to be the bad set, and let $\mathbbm{1}_B$ be the projector onto the corresponding eigenvectors. For each $y \in B$, and as $k$ is even,
   \begin{equation}
       \bra{y} \calP_{k+2} \ket{y} = \sum_{t = 1}^n |\wt{j}_t(y)|^{k+2} > \delta^{k+2} \implies \mathbbm{1}_B \preceq \delta^{-(k+2)}\calP_{k+2}. 
   \end{equation}

   We will use this to bound $\mathrm{bad}(G)$, i.e., $\mu(B)$.
   \begin{align}
       \mathrm{bad}(G) &= \mu(B) = \Tr\left(
            \frac{\ol{\rho}_a^{(n)} + \ol{\rho}_b^{(n)}}{2} \cdot \mathbbm{1}_B
       \right)
       \\& \leq \frac{\delta^{-(k+2)}}{2} \sum_{e \in \{a,b\}} \Tr(\ol{\rho}_e^{(n)} \calP_{k + 2})
       \\&\leq C_k \frac{n^{\frac{k}{2} + 2}}{d^{k+2}},
   \end{align}
   for some $C_k > 0$, where we used \Cref{lem:Jucys--Murphy-large-moments} in the last inequality and absorbed the $\delta$-factors into $C_k$. Consequently, $\mathrm{bad}(G) = o(1)$ whenever $n = o(d^{2 - \gamma})$.
   
    It remains to bound the second term in \Cref{lem:triangular-upper}, i.e., $\mathrm{dev}_{\lambda,G}$. Now, by \Cref{def:normalize}, the log ratio satisfies (on the subspace corresponding to $G$)
    \begin{equation}
        H \coloneqq \log(\ol{\rho}_b^{(n)}) -  \log(\ol{\rho}_a^{(n)}) = \lambda_0 \mathbbm{1} + \sum_{t = 1}^n (\log(f_b(\wt{J}_t)) - \log(f_a(\wt{J}_t))),
    \end{equation}
    for some fixed scalar $\lambda_0$. Define 
    \begin{equation}
        R \coloneqq H - \lambda_0 \mathbbm{1} - \beta_k\calP_k - \beta_{k + 1}\calP_{k+1};
    \end{equation}
    then, by \Cref{prop:not}, we have
    \begin{equation}
        |R \cdot \mathbbm{1}_G| \preceq c_k \calP_{k+2}, \qquad \textnormal{as $|\wt{j}_t(y)| \leq \delta$ on $G$,}
    \end{equation}
    for some $c_k > 0$ depending only on $k$. As in \Cref{lem:Jucys--Murphy-small-moments}, we write $\calP_j = \eta_j \mathbbm{1} + \calP_j^\circ$; then, we choose $\lambda$ such that
    \begin{equation}
        H - \lambda \mathbbm{1} = \beta_k \calP_k^\circ + \beta_{k+1} \calP_{k+1}^\circ + R.
    \end{equation}
    Now, we can write
    \begin{align}
       \mathrm{dev}_{\lambda,G} &= \Ex_\mu \left[
            \left(
                \log \frac{q_y}{p_y} - \lambda
            \right)^2 \mathbbm{1}_G
        \right] = \frac12 \sum_{e \in \{a,b\}} \Tr(\ol{\rho}_e^{(n)} (H - \lambda \mathbbm{1})^2 \mathbbm{1}_G)
        \\& \leq \frac32 \sum_{e \in \{a,b\}} \beta_k^2 \Tr(\ol{\rho}_e^{(n)} (\calP_{k}^\circ)^2) + \beta_{k+1}^2 \Tr(\ol{\rho}_e^{(n)} (\calP_{k+1}^\circ)^2) + \Tr(\ol{\rho}_e^{(n)} (R \cdot \mathbbm{1}_G)^2), \label{eq:nu}
    \end{align}
    where we use $(A+B+C)^2 \preceq 3(A^2 + B^2 + C^2)$ for Hermitian $A,B,C$. As $k$ is even, $\vartheta_k = \vartheta_{k+1} = 2 - \gamma$, and thus, by \Cref{lem:Jucys--Murphy-small-moments}, the first two terms in \Cref{eq:nu} are $o(1)$ for $n = o(d^{2 - \gamma})$.

 Moreover, 
the inequality $(\sum_{i=1}^n a_i)^2\leq n\sum_{i=1}^n a_i^2$ holds for any real numbers $a_i$, and applying it to the eigenvalues of $\left( \sum_{t = 1}^n \wt{J}_t^{k+2}\right)^2$, which commutes with  with $(R \cdot \mathbbm{1}_G)^2$, we can write
    \begin{equation}
        (R \cdot \mathbbm{1}_G)^2 \preceq c_k^2 (\sum_{t = 1}^n \wt{J}_t^{k+2})^2 \preceq c_k^2 \cdot n \calP_{2k+4}.
    \end{equation}
    Thus, by \Cref{lem:Jucys--Murphy-large-moments}, for each $e \in \{a,b\}$,
    \begin{equation}
        \Tr(\ol{\rho}_e^{(n)} (R \cdot \mathbbm{1}_G)^2) \leq O\left(
            \frac{n^{k+4}}{d^{2k+4}}
        \right) = o(1),
    \end{equation}
    for $n = o(d^{2 - \gamma})$. 
    
    We have shown that both $\mathrm{dev}_{\lambda,G}$ and $\mathrm{bad}(G)$ are $o(1)$ for $n = o(d^{2 - \gamma})$; with \Cref{eq:dtv-triangular-2} and \Cref{lem:triangular-upper}, this concludes the proof.
\end{proof}

It remains now to prove \Cref{lem:Jucys--Murphy-small-moments,lem:Jucys--Murphy-large-moments}; first, we will prove some necessary combinatorial lemmas.

\subsection{Useful Lemmas}
\label{sec:useful-lemmas}

In this section, we provide several combinatorial lemmas that will be important for our proofs of \Cref{lem:Jucys--Murphy-small-moments,lem:Jucys--Murphy-large-moments}. For the rest of this section, when describing permutations, we may omit $1$-cycles and only consider non-fixed cycle types. For any non-fixed cycle type $\mu$, let $C_\mu$ denote the conjugacy class of all permutations in $S_n$ corresponding to $\mu$. Note that all permutations in $C_\mu$ have the same support size, number of non-trivial cycles, and Cayley length; we will thus denote the first two by $s(\mu)$ and $c(\mu)$ respectively and extend notation to denote the latter by $|\mu|$. We also define the class sum, $K_\mu \coloneqq\sum_{\pi \in C_\mu} \pi$. 

We begin with bounds on the Cayley lengths of products of transpositions, which arise naturally in the study of high-order moments of Jucys--Murphy elements.

\begin{lemma}
    \label{lem:length-of-swap-products}
    Consider a product of $r$ transpositions with a common center $t$, i.e., a permutation of the form $\pi = (i_1 t) \dots (i_r t)$, for $t \neq i_j, j \in [r]$. Let $q$ be the number of unique elements in the multiset $\{i_1, \dots, i_r\}$. Then,
    \begin{equation}
        2q \leq r + |\pi|.
    \end{equation}
\end{lemma}
\begin{proof}
Let $A$ be the set of unique indices included in the product of transpositions above, including the center $t$. By the hypothesis of the lemma, $|A| = q+1$. 

Let us start with the identity permutation on $A$ and multiply each transposition in the product in order. Each such transposition will either merge two distinct cycles or split a cycle into two. Let $M$ denote the number of merges and $S$ denote the number of splits. Trivially, we have $r = M + S$. Let $c$ be the number of cycles in $\pi_{|A}$, i.e., $\pi$ restricted to the set $A$. Then, 
\begin{equation} 
    c = q+1 - M + S,
\end{equation}
as we start with $q+1$ cycles, each merge reduces the number of cycles, and each split increases it. We can thus write
\begin{equation} \label{eq:length-swap-products-1}
    |\pi| = q+1 - c = M-S.
\end{equation}
Thus, we have
\begin{equation}
    r = |\pi| + 2S. \label{eq:length-swap-products-2}
\end{equation}
Now, in the final permutation $\pi_{|A}$, only one of the $c$ cycles contains the center $t$; however, each element of $A$, at some point, was in the same cycle as $t$. Consequently, each cycle not containing $t$ must have split off from it at some point, and so,
\begin{equation}
    S \geq c-1.
\end{equation}
Thus, by \Cref{eq:length-swap-products-1,eq:length-swap-products-2}
\begin{equation}
    r \geq |\pi| + 2(c-1) = 2q - |\pi|,
\end{equation}
proving the lemma.
\end{proof}

Next, we will show that the power-sums of Jucys--Murphy elements admit a class decomposition and prove bounds on the associated coefficients.

\begin{lemma}[Power-sum class decomposition]
\label{lem:power-sum-conjugacy-decomposition}
Let $m \geq 0$ be a constant integer. Then, the $m$th Jucys--Murphy power sum admits the following class decomposition.
\begin{equation}
    d^m \calP_m = A_{m,\emptyset} \Id + \sum_{\mu \neq \emptyset} A_{m,\mu} K_\mu,
\end{equation}
where the coefficient $A_{m,\mu} = 0$ whenever $|\mu| \not\equiv m \textnormal{ (mod } 2)$ or $|\mu| > m$; otherwise 
\begin{equation}
    A_{m,\mu} \leq O(n^{\delta(\mu)}), \quad \delta(\mu) \coloneqq\frac{m+2-|\mu|-2c(\mu)}{2}.
\end{equation}
\end{lemma}

We note that exact expressions of the coefficients $A_{m,\mu}$ have appeared in prior work (see e.g., \cite[Theorem 6.4]{lassalle2013class}); however, we include a self-contained proof of our asymptotic bounds for completeness.

\begin{proof}
    We can rewrite the power sum as $d^m \calP_m = \sum_{t = 1}^n J_t^m$, i.e., a symmetric polynomial in the Jucys--Murphy elements. Recall that symmetric polynomials in the Jucys--Murphy elements are in the center of the permutation algebra, which is generated by class functions $K_\mu$. This shows the existence of the class decomposition. 
    
    We will now prove the desired properties of the coefficients $A_{m,\mu}$. Let us write $d^m \calP_m = \sum_{\pi \in S_n} c_\pi \pi$. By the class decomposition, for any fixed permutation $\pi \in C_\mu$, we have $c_\pi = A_{m,\mu}$. Now, we can write
    \begin{equation}
        \sum_{t = 1}^n J_t^m = \sum_{t = 1}^n \sum_{i_1, \dots, i_m < t} (i_1 t) \dots (i_m t), \label{eq:class-decomposition-1}
    \end{equation}
    i.e., a sum of products of transpositions. Thus, the coefficient $c_\pi$ is the number of such products of transpositions that yield the permutation $\pi$. 

    Note that such products cannot produce a permutation with $|\pi| > m$, by definition of the Cayley length. Further, by \Cref{eq:length-swap-products-2} in the proof of \Cref{lem:length-of-swap-products}, we can only obtain permutations whose Cayley length $|\pi|$ has the same parity as $m$. Thus, $A_{m,\mu} = 0$ whenever $|\mu| \not\equiv m \textnormal{ (mod } 2)$ or $|\mu| > m$, as claimed. 
    
    It remains to prove the claimed upper bound on the non-zero coefficients $A_{m,\mu}$. For a fixed permutation $\pi \in C_\mu$, as $A_{m,\mu} = c_\pi$, we wish to bound the number of choices of the indices $i_1, \dots, i_m, t$ in \Cref{eq:class-decomposition-1} that yield $\pi$. Let $A = \{t\} \cup \{i_1\} \dots \cup \{i_m\}$ be the set of indices acted upon by the product of transpositions. For these indices to correspond to $\pi$, it must be the case that $\supp(\pi) \subseteq A$. Thus, to prove our desired bound, we will aim to understand the number of remaining elements, i.e., $|A \setminus \supp(\pi)|$. 

    Now, consider the restricted permutation $\pi_{|A}$ that describes the effect of $\pi$ on the index set $A$. Let $b$ be the number of cycles in $\pi_{|A}$, and let $q = |A| - 1$. Then, by \Cref{lem:length-of-swap-products}, we have
    \begin{equation}
        2q \leq m + |\pi| \quad \implies \quad m \geq |\pi| + 2(b-1), 
    \end{equation}
    as $|\pi| = |\pi_{|A}| = q+1 - b$. Note that $b = c(\mu) +  |A \setminus \supp(\pi)|$, as each element fixed by $\pi_{|A}$ contributes one cycle. Thus,
    \begin{equation}
        |A \setminus \supp(\pi)| \leq \frac{m - |\pi| -2c(\mu) + 2}{2} = \delta(\mu),
    \end{equation}
    as $|\mu| = |\pi|$ for all $\pi$ with non-trivial cycle type $\mu$. 

    Recall that $A$ must contain $\supp(\pi)$. Thus, one can choose $i_1, \dots, i_m, t$ by first choosing the remaining elements, i.e., $A \setminus \supp(\pi)$, and then assigning valid values to $i_1, \dots, i_m,t$ from $A$. Given $A$, the number of such assignments is upper bounded by a finite function of $m$, which is a constant. Thus, the total number of choices is dominated by those of $A \setminus \supp(\pi)$, which is    
    at most $O(n^{|A \setminus \supp(\pi)|}) \leq O(n^{\delta(\mu)})$, as claimed.
\end{proof}

\begin{lemma}
\label{lem:support-pair-estimates}
    For a non-trivial cycle type $\mu \neq \emptyset$ and support size $s = s(\mu)$, the number of ordered pairs $(\sigma,\tau) \in C_\mu\times C_\mu $ with
    \begin{equation}
        |\supp(\sigma) \setminus \supp(\tau)| = |\supp(\tau) \setminus \supp(\sigma)| = u
    \end{equation} 
    is $O_s(n^{s+u})$. For any such pair, if $\sigma \neq \tau^{-1}$, then
    \begin{equation}
        |\sigma\tau| \geq \max\{2, 2u - 2c(\mu)\}.
    \end{equation}
    Moreover, if $u = c(\mu) + 1$ and $|\mu|$ is even, then
    \begin{equation}
        |\sigma\tau| \geq 4.
    \end{equation}
\end{lemma}

Before proving the above lemma, let us state some necessary notation.

\begin{definition}
    Given two permutations $\sigma, \tau \in S_n$, let $\langle \sigma,\tau \rangle$ be the group of permutations generated by them. For any index $i \in [n]$, its orbit under $\langle \sigma,\tau \rangle$ is the set of elements that can be reached under any number of (potentially alternating) applications of $\sigma$ and $\tau$, including the starting index $i$ itself.
\end{definition}

Alternatively, we will view these orbits as the connected components of an appropriate graph.

\begin{definition}
\label{def:orbit-graph}
    Given permutations $\sigma, \tau \in S_n$, we define the graph $G_{\sigma,\tau}$ on vertices $1, \dots, n$ to contain an undirected edge $(i, j)$ whenever $\sigma(i) = j$ or $\tau(i) = j$. Traversal along the edges of this graph corresponds to the actions of $\sigma, \tau$. Thus, the number of orbits of $\langle \sigma, \tau \rangle$ is precisely the number of connected components of $G_{\sigma,\tau}$.
\end{definition}

The proof of \Cref{lem:support-pair-estimates} will rely on the following standard fact about the number of orbits of such groups $\langle \sigma, \tau\rangle$. We include the proof for completeness.

\begin{fact}
    \label{fact:genus-inequality}
    Let $\sigma, \tau \in S_n$. Let $o$ be the number of orbits of $\langle \sigma, \tau\rangle$. Then,
    \begin{equation}
        |\sigma| + |\tau| + |\sigma\tau| \geq 2(n-o).
    \end{equation}
\end{fact}

\begin{proof}
    Let $\tau = t_1 \dots t_m$ be a decomposition of $\tau$ into the minimal number of transpositions $m = |\tau|$. We will add these transpositions to $\sigma$ one at a time, i.e., we will consider permutations $\pi_j \coloneqq\sigma t_1 \dots t_j$ for $1 \leq j \leq m$, and let $\pi_0 \coloneqq\sigma$. Now, adding the transposition $t_j$ to $\pi_{j-1}$ will either merge two cycles of $\pi_{j-1}$ into one, or split one cycle into two. Let $M$ and $S$ be the number of merges and splits, respectively, when going from $\sigma$ to $\sigma\tau$ in this manner, with $M + S = m$. We will also have
    \begin{equation} \label{eq:genus-inequality-1}
        \numcyc(\sigma\tau) = \numcyc(\sigma) - M + S.
    \end{equation}
    
    Now, let us define the graph $G_0 \coloneqq G_\sigma$ to be the connectivity graph of $\sigma$, i.e., the graph on vertices $1,\dots,n$ with vertices $i,j$ connected by an undirected edge whenever $\sigma(i) = j$. For each transposition $t_j = (a_j b_j)$, we define the corresponding transposition edge to be the one connecting the indices $a_j$ and $b_j$. We will recursively define graphs $G_1, \dots, G_m$ by adding a $t_j$-edge to $G_{j-1}$.
    
    The resultant graph $G_m$ allows traversal along $\sigma$ and any transposition of $\tau$; as these transpositions generate $\tau$ itself, one can actually perform any $\langle \sigma, \tau\rangle$-traversal along the graph. Requiring $t_1 \dots t_m$ to be a minimal decomposition of $\tau$ is crucial here to ensure that this graph \emph{only} allows  $\langle \sigma,\tau \rangle$-traversal and nothing more. For instance, if $\tau$ consisted of only one cycle, say, $(1 \dots k)$, then a minimal decomposition would have $k-1$ transpositions, each corresponding exactly to a traversal permissible under $\tau$ or $\tau^{-1}$; this generalizes to generic $\tau$ with multiple cycles by applying the same argument within each cycle and noting that a minimal decomposition contains no transpositions across cycles. Thus, the graph $G_m$ allows precisely the traversals under $\langle \sigma, \tau\rangle$, implying its number of connected components is the number of orbits of $\langle \sigma, \tau \rangle$, i.e., $o$.
    
    We will now show that the number of cycle-merges $M$ is at least the number of merges of components. First, as an invariant with respect to $j$, we will maintain that each cycle of $\pi_j$ is entirely contained in some component of $G_j$. Clearly, this is the case for the base graph $G_0$; it suffices to show that adding an edge maintains this invariant. Consider two cases when adding the edge for $t_j = (a_j b_j)$.

    First, imagine that $a_j$ and $b_j$ are already in the same component. Then, adding this edge does not change the connectivity of the graph. Further, this maintains our invariant, as the elements of the new cycle(s) of $\pi_j$ will still be contained in this component if those of $\pi_{j-1}$ were.

    Alternatively, imagine that $a_j$ and $b_j$ belonged to different components. The new edge will result in a merge of these two components. Further, by our invariant, $a_j$ and $b_j$ must have been in different cycles of $\pi_{j-1}$. Consequently, adding the transposition $t_j$ will merge these two cycles. This also maintains our invariant, as all elements of the merged cycle lie in the new merged component. 

    As we have shown, the invariant is maintained throughout, and the only case in which one obtains a merge of components also results in a merge of cycles. Consequently, 
    \begin{equation}
        M \geq \numcyc(\sigma) - o. \label{eq:genus-inequality-2}
    \end{equation}
    Using \Cref{eq:genus-inequality-1}, we write
    \begin{equation}
        \numcyc(\sigma\tau) = \numcyc(\sigma) - M + S = \numcyc(\sigma) + |\tau| - 2M,
    \end{equation}
    as $M + S = |\tau|$. Now, in terms of Cayley length,
    \begin{align}
        |\sigma\tau| &= n - \numcyc(\sigma\tau) = n - \numcyc(\sigma) - |\tau| + 2M 
        \\& \geq n + \numcyc(\sigma) - 2o - |\tau|
        \\&= 2n - 2o - |\sigma| - |\tau|,
    \end{align}
    where the inequality used \Cref{eq:genus-inequality-2}, thus proving the lemma.
\end{proof}

We now prove \Cref{lem:support-pair-estimates}.

\begin{proof}[Proof of \Cref{lem:support-pair-estimates}]
    To construct such an ordered pair, one can start by choosing the $s+u$ indices in $\supp(\sigma) \cup \supp(\tau)$; the number of permutations that can be constructed given these indices is at most a finite function of $s$. Thus, the number of ordered pairs is dominated by the number of choices for $\supp(\sigma) \cup \supp(\tau)$, yielding the $O_s(n^{s + u})$ bound. 

    We will now prove lower bounds on the Cayley length of $|\sigma\tau|$ whenever $\sigma \neq \tau^{-1}$. First, we will show that $|\sigma\tau| \geq 2$; for this, it suffices to show that $|\sigma\tau|$ is even, as the Cayley length is always non-negative and cannot be $0$ when $\sigma \neq \tau^{-1}$. 
    
    Without loss of generality, let us restrict the actions of $\sigma$ and $\tau$ to $\supp(\sigma) \cup \supp(\tau)$, and let $N = |\supp(\sigma) \cup \supp(\tau)| = s+u$ be the size of this union. Adopting the merging and splitting framework in the proof of \Cref{lem:length-of-swap-products}, let us incorporate each of the $|\tau|$ transpositions of $\tau$ into $\sigma$ one at a time and in order; suppose this involves $M$ merges of two distinct cycles and $S$ splits of a cycle into two, with $M+S = |\tau|$. Then,
    \begin{equation}
        \numcyc(\sigma\tau) = \numcyc(\sigma) - M + S = \numcyc(\sigma) - |\tau| + 2S.
    \end{equation}
    However, $|\sigma| = |\tau|$, and so
    \begin{equation}
        \numcyc(\sigma\tau) = \numcyc(\sigma) - (N - \numcyc(\sigma)) + 2S = 2\numcyc(\sigma) + 2S - N.
    \end{equation}
    We thus have
    \begin{equation}
        |\sigma\tau| = N - \numcyc(\sigma\tau) = 2N - 2S - 2\numcyc(\sigma) = 2|\sigma| - 2S,
    \end{equation}
    which is even, as desired. Thus, we have shown that $|\sigma\tau| \geq 2$.

    Next, we will show that $|\sigma\tau| \geq 2u - 2c(\mu)$. We will apply \Cref{fact:genus-inequality} to our pair of permutations. Note that the union $\supp(\sigma) \cup \supp(\tau)$ is entirely covered by the non-trivial cycles of $\sigma$ and $\tau$, implying that the number of orbits $o$ of $\langle \sigma, \tau\rangle$ is at most $2c(\mu)$. Further, we can write $N = s + u = |\mu| + c(\mu) + u$. Consequently, by \Cref{fact:genus-inequality}, 
    \begin{equation}
        |\sigma\tau| + |\sigma| + |\tau| \geq 2(|\mu| + c(\mu) + u - 2c(\mu)) \implies |\sigma\tau| \geq 2u - 2c(\mu),
    \end{equation}
    as $|\sigma| = |\tau| = |\mu|$.

    Lastly, it remains to handle the special case where $u = c(\mu) + 1$ and $|\mu|$ is even. As $\mu \neq \emptyset$, the latter assumption implies $|\mu| \geq 2$. Further, the intersection $\supp(\sigma) \cap \supp(\tau)$ has size
    \begin{equation}
        |\supp(\sigma) \cap \supp(\tau)| = s - u = |\mu| + c(\mu) - c(\mu) - 1 = |\mu| - 1 \geq 1.
    \end{equation}
    Thus, $\supp(\sigma)$ and $\supp(\tau)$ share at least one element. This element belongs to one non-trivial cycle each of $\sigma$ and $\tau$, implying that the corresponding elements in both cycles share an orbit. Thus, we obtain the refined bound $o \leq 2c(\mu) - 1$. Applying \Cref{fact:genus-inequality} again, we get
    \begin{equation}
        |\sigma\tau| + |\sigma| + |\tau| \geq 2(|\mu| + c(\mu) + u - 2c(\mu) + 1) = 2|\mu| + 4 \implies |\sigma\tau| \geq 4,
    \end{equation}
    where we used $u = c(\mu) + 1$ and $|\sigma| = |\tau| = |\mu|$. We have thus shown all desired lower bounds on the Cayley length $|\sigma\tau|$.
\end{proof}

Lastly, we will require the following lemma upper-bounding the inner products of our two mixtures with permutation operators.
\begin{lemma}
    \label{lem:pointwise-decay}
    For $e \in \{a,b\}$, and any permutation $\pi \in S_n$,
    \begin{equation}
        \Tr(\pi \ol{\rho}_e^{(n)}) \leq \left(\frac{C_0}{d}\right)^{|\pi|},
    \end{equation}
    for some $C_0 > 0$ that depends only on $k$.
\end{lemma}

\begin{proof}
    Recall that averaged $n$-fold state can be written as
    \begin{equation}
        \ol{\rho}_e^{(n)} = \frac{\Ex \bX_e^{\otimes n}}{\Ex \Tr (\bX_e)^n},
    \end{equation}
    where $\bX_e$ is drawn from $\PRP_d(e)$. By \Cref{fact:perm-product},
    \begin{equation}
        \Tr(\pi \bX_e^{\ot n}) = \prod_{\tau \in \cyc(\pi)} \Tr(\bX_e^{\mathrm{len}(\tau)}) \leq \Tr(\bX_e)^{\numcyc(\pi)},
    \end{equation}
    as $0 \preceq \bX_e \preceq I$. Thus,
    \begin{equation}
        \Tr(\pi \ol{\rho}_e^{(n)}) \leq \frac{\Ex \Tr(\bX_e)^{\numcyc(\pi)}}{\Ex \Tr(\bX_e)^{n}} \leq (\Ex \Tr(\bX_e))^{\numcyc(\pi) - n} = (\Ex \Tr(\bX_e))^{-|\pi|}. \label{eq:pointwise-decay-1}
    \end{equation}
    To see why the second inequality holds, note that for $0 \leq k \leq n$, by the monotonicity of $L_p$-norms, $\Ex [\Tr(\bX_e) ^k] \leq \Ex[\Tr(\bX_e)^n]^{k/n}$. Applying this inequality twice with $k = \numcyc(\pi)$ and $k = 1$, we get
    \begin{equation}
        \Ex [\Tr(\bX_e)^{\numcyc(\pi)}] \Ex[\Tr(\bX_e)]^{n-\numcyc(\pi)} \leq \Ex[\Tr(\bX_e)^n]^{\numcyc(\pi)/n} \Ex[\Tr(\bX_e)^n]^{\frac{n-\numcyc(\pi)}{n}} = \Ex[\Tr(\bX_e)^n],
    \end{equation}
    yielding the second inequality in \Cref{eq:pointwise-decay-1} after dividing by $\Ex[\Tr(\bX_e)^n] \Ex[\Tr(\bX_e)]^{n-\numcyc(\pi)}$.
    
    Now, by \Cref{eq:pointwise-decay-1}, it suffices to show that $\Ex \Tr(\bX_e) \geq C_0^{-1} d$. Using \Cref{eq:average1}, we have
    \begin{equation}
        \Ex \Tr(\bX_e)= \nu_e d. 
    \end{equation}
    This proves the lemma for $C_0 = \max\{\nu_a^{-1},\nu_b^{-1}\}$, where $\nu_e^{-1} = e_1\dots e_K$ depends only on $k$.

\end{proof}

\subsection{Proofs of \Cref{lem:Jucys--Murphy-small-moments,lem:Jucys--Murphy-large-moments}}

Given \Cref{lem:power-sum-conjugacy-decomposition,lem:support-pair-estimates}, we can finally prove \Cref{lem:Jucys--Murphy-small-moments}.

\begin{proof}[Proof of \Cref{lem:Jucys--Murphy-small-moments}]
    \Cref{lem:power-sum-conjugacy-decomposition} implies
    \begin{equation}
        \calP_m^\circ = \sum_{\mu \neq \emptyset} \frac{A_{m,\mu}}{d^m} K_\mu,
    \end{equation}
    where the sum is over all possible non-trivial cycle types except for $\mu = \emptyset$. As the only non-zero coefficients correspond to cycle types with $|\mu| \leq m$, we must have $s(\mu) \leq 2m$. Consequently, the number of summands above is bounded by some finite function of $m$. We thus have
    \begin{equation}
        \Tr (\ol{\rho}_e^{(n)} (\calP_m^\circ)^2) \leq O_m \left(\sum_{\mu \neq \emptyset} \frac{A_{m,\mu}^2}{d^{2m}} \Tr(\ol{\rho}_e^{(n)} K_\mu^2) \right),
    \end{equation}
    where we use $(A_1 + \dots + A_N)^2 \leq N(A_1^2 + \dots + A_N^2)$ for Hermitian $A_1, \dots, A_N$.
    
    Now, for a fixed cycle type $\mu$, we write
    \begin{equation}
        \Tr (\ol{\rho}_e^{(n)} K_\mu^2) =  \sum_{\sigma,\tau \in C_\mu} \Tr (\ol{\rho}_e^{(n)}\sigma\tau) \leq \sum_{\sigma,\tau \in C_\mu} (C_0/d)^{|\sigma\tau|},
        \label{eq:variance-bound-1}
    \end{equation}
    where the upper bound is from \Cref{lem:pointwise-decay}. Keeping \Cref{lem:support-pair-estimates} in mind, we will further group the summation over ordered pairs $(\sigma,\tau)$ by the size of their difference, i.e., $|\supp(\sigma) \setminus \supp(\tau)|$. Thus, by~\eqref{eq:variance-bound-1},
    \begin{align}
        \frac{A_{m,\mu}^2}{d^{2m}} \Tr (\ol{\rho}_e^{(n)} K_\mu^2) \leq \sum_{u = 0}^{s(\mu)} \sum_{\substack{\sigma,\tau \in C_\mu\\ |\supp(\sigma) \setminus \supp(\tau)| = u}} \frac{A_{m,\mu}^2}{d^{2m}} (C_0/d)^{|\sigma\tau|}.
        \label{eq:variance-bound-2}
    \end{align}
    Let us split the above summation into terms with $\sigma = \tau^{-1}$ and those where this is not the case. Note that in the former case, both permutations have the same supports, corresponding to $u = 0$ and $|\sigma\tau| = 0$. Further, in the latter case, for fixed $\mu,u$, let $\ell_{\mu,u}$ be the best possible lower bound for $|\sigma\tau|$ implied by \Cref{lem:support-pair-estimates}. Then, by \Cref{lem:power-sum-conjugacy-decomposition} and the first part of \Cref{lem:support-pair-estimates}, we have
    \begin{equation}
        \eqref{eq:variance-bound-2} \leq O_m\left(\frac{n^{s(\mu) + 2\delta(\mu)}}{d^{2m}}\right) + \sum_{u = 0}^{s(\mu)} n^{s(\mu) + u + 2\delta(\mu)} \cdot O(d^{-1})^{2m + \ell_{\mu,u}}.
    \end{equation}
    Consequently, 
    \begin{align}
        \Tr (\ol{\rho}_e^{(n)} (\calP_m^\circ)^2) &\leq \sum_{\mu \neq \emptyset, A_{m,\mu} \neq 0} \left[ O_m\left(\frac{n^{s(\mu) + 2\delta(\mu)}}{d^{2m}}\right) + \sum_{u = 0}^{s(\mu)} n^{s(\mu) + u + 2\delta(\mu)} \cdot O_m(d^{-1})^{2m + \ell_{\mu,u}} \right]\\&\leq \max_{\substack{\mu \neq \emptyset, A_{m,\mu} \neq 0 \\ u \in \{0, \dots, s(\mu)\}}} O_m\left(\frac{n^{s(\mu) + 2\delta(\mu)}}{d^{2m}} + \frac{n^{s(\mu) + u + 2\delta(\mu)}}{d^{2m + \ell_{\mu,u}}}\right),\label{eq:variance-bound-3}
    \end{align}
    where the second inequality used that the number of summands above is at most a finite function of $m$ and can thus be absorbed into the $O_m$. Thus, it suffices to show that the two terms above are $o(1)$ for all choices of $\mu, u$, and any $m = O(1)$. 

    First, for the $u$-independent term, note that $s(\mu) + 2\delta(\mu) = m+2-c(\mu)$, as $s(\mu) = |\mu| + c(\mu)$. As $\mu \neq \emptyset,$ we must have $c(\mu) \geq 1$, i.e., $s(\mu) + 2\delta(\mu) \leq m+1$. Consequently, for the first term in \Cref{eq:variance-bound-3}, we have
    \begin{equation}
        O_m\left(\frac{n^{s(\mu) + 2\delta(\mu)}}{d^{2m}}\right) \leq O_m\left(\frac{n^{m+1}}{d^{2m}}\right). 
    \end{equation}
    Now, for constant $m$, the above is at most $o(1)$ unless $n \geq \Omega(d^{\frac{2m}{m+1}}) \geq \Omega(d^{\vartheta_m})$, as $\frac{2m}{m+1} \geq \vartheta_m$ whenever $m \geq 2$. Thus, under the hypothesis of the lemma, the above term is $o(1)$.

    It remains now to show that the latter term in \Cref{eq:variance-bound-3} is $o(1)$ whenever $n = o(d^{\vartheta_m})$ and for all valid choices of $\mu,u$ and any constant $m$. We will do this by splitting into three cases depending on $u$ and applying \Cref{lem:support-pair-estimates} to obtain $\ell_{\mu,u}$ in each case. Recall that each case only considers contributions from pairs with $\sigma \neq \tau^{-1}$. 

    \paragraph{Case 1: $u \leq c(\mu)$.} For the exponent of $n$, we have $s(\mu) + u + 2\delta(\mu) = m + 2 + u - c(\mu) \leq m+2$. For the exponent of $d$, \Cref{lem:support-pair-estimates} tells us $\ell_{\mu,u} \geq 2$, and so 
    \begin{equation}
        n^{s(\mu) + u + 2\delta(\mu)} \cdot O_m(d^{-1})^{2m + \ell_{\mu,u}} \leq n^{m+2} \cdot O_m(d^{-1})^{2m+2}.
    \end{equation}
    For constant $m$, the above is $o(1)$ unless $n \geq \Omega(d^{\frac{2m+2}{m+2}}) > \Omega(d^{\vartheta_m})$, as $\frac{2m+2}{m+2} > \vartheta_m$ for all $m \geq 2$. By the hypothesis of the lemma, this variance term is thus $o(1)$.

    \paragraph{Case 2: $u = c(\mu) + 1$.} Here, we have $s(\mu) + u + 2\delta(\mu) = m+3$. Let us first consider the case when $m$ is even. Then, by \Cref{lem:power-sum-conjugacy-decomposition}, the Cayley length $|\mu|$ must also be even. By the ``moreover'' part of \Cref{lem:support-pair-estimates}, we obtain $\ell_{\mu,u} \geq 4$. Thus, 
    \begin{equation}
        n^{s(\mu) + u + 2\delta(\mu)} \cdot O_m(d^{-1})^{2m + \ell_{\mu,u}} \leq n^{m+3} \cdot O_m(d^{-1})^{2m+4}.
    \end{equation}
    As before, the above term is $o(1)$ unless $n \geq \Omega(d^{\frac{2m+4}{m+3}}) > \Omega(d^{\vartheta_m})$, as $\frac{2m+4}{m+3} > \vartheta_m$ for all even $m \geq 2$. 
    
    Now, for odd $m$, it suffices to take the weaker bound of $\ell_{\mu,u} \geq 2$ from \Cref{lem:support-pair-estimates}, implying
    \begin{equation}
        n^{s(\mu) + u + 2\delta(\mu)} \cdot O_m(d^{-1})^{2m + \ell_{\mu,u}} \leq n^{m+3} \cdot O_m(d^{-1})^{2m+2}.
    \end{equation}
    Again, this is $o(1)$ unless $n \geq \Omega(d^{\frac{2m+2}{m+3}}) = \Omega(d^{\vartheta_m})$, as one can verify that $\frac{2m+2}{m+3} = \vartheta_m$ for all odd $m$. 
    
    Thus, under the hypothesis of the lemma, for both odd and even $m$, the variance contribution is $o(1)$ whenever $u = c(\mu) + 1$.

    \paragraph{Case 3: $u \geq c(\mu) + 2$.} Here, we will just write $s(\mu) + u + 2\delta(\mu) = m+2 + u - c(\mu)$, and, by \Cref{lem:support-pair-estimates}, $\ell_{\mu,u} \geq 2u - 2c(\mu)$. Consequently,  
    \begin{equation}
        n^{s(\mu) + u + 2\delta(\mu)} \cdot O_m(d^{-1})^{2m + \ell_{\mu,u}} \leq n^{m+2 + u -c(\mu)} \cdot O_m(d^{-1})^{2m+2u-2c(\mu)}.
    \end{equation}
    Now, this is again $o(1)$ unless $n \geq \Omega(d^{\frac{2m+2u-2c(\mu)}{m+2+u-c(\mu)}})$. However, this exponent satisfies
    \begin{equation}
        \frac{2m+2u-2c(\mu)}{m+2+u-c(\mu)} = 2 - \frac{4}{m+2 + u - c(\mu)} \geq 2 - \frac{4}{m+4} \geq 2-\frac{2}{\lfloor m/2 \rfloor + 2} = \beta_m,
    \end{equation}
    where the first inequality used that, in this case, $u \geq c(\mu) +2$. Consequently, for $n = o(d^{\vartheta_m})$, the variance contribution in this case is again $o(1)$. 
\end{proof}

Lastly, we prove \Cref{lem:Jucys--Murphy-large-moments}.

\begin{proof}[Proof of \Cref{lem:Jucys--Murphy-large-moments}]
    We will first rewrite the normalized power-sums in terms of products of swaps.
    \begin{equation}
        d^{2r} \calP_{2r} = \sum_{t = 1}^n J_t^{2r} = \sum_{t = 1}^n \sum_{i_1, \dots, i_{2r}} (i_1 t) \dots (i_{2r} t). \label{eq:even-moments-1}
    \end{equation}
    We will group the products of transpositions above by the Cayley length $\ell$ of the resulting permutation. For any such product, let $q$ denote the number of unique elements in the multiset $\{i_1, \dots, i_{2r}\}$. As shown in \Cref{lem:length-of-swap-products}, this will satisfy $q \leq r + \ell/2$. 
    
    Now, such a product can be chosen by first picking the index $t$, then the $q$ unique elements in $\{i_1,\dots, i_{2r}\}$, and then appropriately ordering these elements. There are $n$ ways to pick $t$, at most $\binom{n-1}{q} \leq n^q$ ways to pick the unique indices, and, crudely, at most $q^{2r} \leq (2r)^{2r}$ ways\footnote{We made no effort to optimize this bound as it depends only on $r$.} to order these indices. In total, there are at most $(2r)^{2r} \cdot n^{1 + r + \ell/2}$ summands in \Cref{eq:even-moments-1} that correspond to permutations of Cayley length $\ell$. Now, \Cref{lem:pointwise-decay} implies
    \begin{equation}
        d^{2r} \Tr(\ol{\rho}_e^{(n)}\calP_{2r}) \leq (2r)^{2r} n^{r+1} \sum_{\ell = 0}^{2r} n^{\ell/2} \cdot (C_0/d)^\ell \leq (2r)^{2r} n^{r+1} \cdot 2,
    \end{equation}
    whenever $n \leq \frac{d^2}{4C_0^2}$. Dividing both sides by $d^{2r}$ proves the lemma with $C_{k,r} = 2 \cdot (2r)^{2r}$ and $c_k = \frac{1}{4C_0^2}$.
\end{proof}

\section{Spectral Separation}
\label{sec:spectral-separation}

Throughout this section, we fix constant $k\in\N$, $K=K(k)$, and $a,b\in\Z_+^K$ as in \Cref{prop:prouhet}.
In this section, all $O(\cdot)$ notation permits constants that may depend on $k$, $a$, and $b$.

\subsection{Deterministic moment separation}\label{sec:det_moment}

In this section, we discuss how the Prouhet–Thue–Morse construction implies matching of the first~$k$ (suitably normalized) expected moments, and a large enough separation for the $(k+1)$th moment. We will show later that these quantities track the typical moments of the tilted ensembles.
For $e\in\{a,b\}$, we denote the normalized expected  $j$th moment as

\begin{equation}
\Psi_j(e)
\coloneqq
\frac{
\E\!\left[\tr(\bX_e^j)\right]
}{
\E[\tr(\bX_e)]^j
}\,,
\label{eq:spectral-moment-proxy}
\end{equation}
where we recall that $\tr(A)\coloneqq d^{-1}\Tr(A)$ denotes the normalized trace. Tracking how $\Psi_j(e)$ changes by increasing $j$ gives the following lemma.

\begin{lemma}
\label{lem:first-spectral-gap}
For $1\leq j\leq k$,
\begin{equation}
\Psi_j(a)-\Psi_j(b)
=
O(d^{-1}).
\label{eq:low-degree-proxy-matching}
\end{equation}
There is also a constant $\kappa_k\neq0$ such that
\begin{equation}
\Psi_{k+1}(a)-\Psi_{k+1}(b)
=
\kappa_k+O(d^{-1}).
\label{eq:first-proxy-gap}
\end{equation}
\end{lemma}

\begin{proof}
Fix $j\leq k+1$ and let $\zeta_j=(1~2~\cdots~j)$.  By
\Cref{eqn:tr},
\begin{equation}
\Tr(\zeta_j X^{\otimes j})=\Tr(X^j).
\end{equation}
Using \Cref{eq:average1} and \Cref{prop:f} we can simplify the $\nu_e^j$ term and write
\begin{equation}
\Psi_j(e)
=
\tr\parens*{\zeta_j \prod_{t=1}^j f_e(\wt{J}_t)}\,.
\end{equation}
We have the following asymptotic bound as a consequence of \Cref{prop:not} and the fact that $f_a(0)=f_b(0)=1$:
\begin{equation}
f_a(z)-f_b(z)=-\beta_{k}z^{k}+O(z^{k+1})\,.\label{eq:spectral-separation-1}
\end{equation}
Since every transposition acts unitarily on $(\C^d)^{\otimes n}$,
the triangle inequality gives
\begin{equation}
\|J_t\|_\infty
=
\left\|
\sum_{i=1}^{t-1}(i~t)
\right\|_\infty
\leq
\sum_{i=1}^{t-1}\|(i~t)\|_\infty
=
t-1.
\label{eq:jm-operator-norm}
\end{equation}
Consequently, $\|\wt{J}_t\|_\infty
\leq
\frac{t-1}{d}$ and, for $t\leq j\leq k+1$, $\|f_e(\wt{J}_t)\|_{\infty}$ is uniformly bounded by a $k$-dependent constant. We can then telescope the product and use submultiplicativity of the operator norm along with \Cref{eq:spectral-separation-1} to obtain:
\begin{equation}
\left\|\prod_{t=1}^jf_a(\wt{J}_t)-\prod_{t=1}^jf_b(\wt{J}_t)\right\|_{\infty}\leq \sum_{s=1}^{j}\left\|\prod_{t=1}^{s-1}f_a(\wt{J}_t)\left(f_a(\wt{J}_s)-f_b(\wt{J}_s)\right)\prod_{t'=s+1}^{j}f_b(\wt{J}_{t'})\right\|_{\infty}\leq O(d^{-k})\,.
\end{equation}
The cycle $\zeta_j$ is unitary, so $\|\zeta_j\|_{1}=d^{j}$, and it follows that
\begin{equation}
|\Psi_j(a)-\Psi_j(b)|\leq \frac{\|\zeta_j\|_{1}\left\|\prod_{t=1}^jf_a(\wt{J}_t)-\prod_{t=1}^jf_b(\wt{J}_t)\right\|_{\infty}}{d}
\leq\frac{d^{j}O(d^{-k})}{d}=O(d^{-1})\,,
\end{equation}
for $j \leq k-1$, proving the first claim.
The same telescoping trick and expansions also show that, for $j=k+1$,
\begin{equation}
\prod_{t=1}^jf_a(\wt{J}_t)-\prod_{t=1}^jf_b(\wt{J}_t)=-\beta_k\sum_{t=1}^{k+1}\wt{J}^k_t+E\,,\qquad \|E\|_{\infty}=O(d^{-k-1})
\end{equation}
We then have 
\begin{equation}
\Psi_{k+1}(a)-\Psi_{k+1}(b)
=-\beta_{k}\sum_{t=1}^{k+1}\tr\parens{\zeta_{k+1}\wt{J}_{t}^{k}}
+O(d^{-1}).
\end{equation}
The cycle $\zeta_{k+1}^{-1}$ can only be realized as a product of $k$ transpositions if such transpositions touch all the $k+1$ points moved by $\zeta_{k+1}$. Each term $\wt{J}_{t}^{k}$ can be expanded as a sum of products of transpositions with center $t$, but only $\wt{J}_{k+1}$ contains enough different transpositions to reconstruct $\zeta_{k+1}^{-1}$, indeed $(1\, k+1)(2\, k+1)\cdots (k\, k+1)=\zeta_{k+1}^{-1}$. It follows that the coefficient of $\zeta_{k+1}^{-1}$ in ${J}_{k+1}^{k}$ is $1$, since any other order would produce a different $k+1$ cycle. Moreover, all the other permutations $\pi$ in the expansion of ${J}^{k}_{t}$ are such that $\zeta_{k+1}\pi$ is not the identity permutation, therefore it has at most $k$ cycles, and $\Tr(\zeta_{k+1}\pi)\leq d^{k}$. Therefore, since the expansion contains only $O(1)$ terms,
\begin{equation}
\sum_{t=1}^{k+1}\tr\parens{\zeta_{k+1}\wt{J}_{t}^{k}}=1+O(d^{-1}),
\end{equation}
and the claim follows.
\end{proof}

\subsection{Probabilistic spectral separation}

The random product of projectors $\bX$ can be seen as a function of $K$ Haar-random unitaries $\bU_{1},\cdots, \bU_{K}$, through $\bPi_i=\bU_iP^{(0)}_{a_i}\bU_i^{\dagger}$, where $P^{(0)}_{a_i}$ is some fixed projector of rank $d/a_i$. Let $F:\mathrm{U}(d)^{\times K}\rightarrow \mathbb{R}$ be a real function that is $\Lambda$-Lipschitz with respect to the $L_2$-sum of Hilbert--Schmidt metrics on $U(d)$, i.e., for any $\{U_i\},\{U'_i\}\in\mathrm{U}(d)^{\times K}$

\begin{equation}
|F(U_1,\cdots, U_K)-F(U'_1,\cdots, U'_K)|\leq \Lambda \sqrt{\sum_{i=1}^{K}\|U_i-U'_i\|^2_{2}}.
\end{equation}
By a standard result on the concentration of Lipschitz functions of independent Haar-random unitaries, see e.g.~\cite[Theorem~5.17]{Meckes_2019}, we have
\begin{equation}\label{eq:concentration-lips}
\Pr(|F(\bU_1,\cdots, \bU_K)-\E F(\bU_1,\cdots, \bU_K))|\geq u)\leq 2e^{-\frac{(d-2)u^2}{24\Lambda^2}}\,.
\end{equation}

In the following, we use non-bold variables with the same functional relations as the corresponding random variables.

\begin{lemma}\label{lem:lipschitz}
The function $F_j(U_1,\cdots, U_K)\coloneqq\tr(X^j)$ is $\left(\frac{4Kj^2}{d}\right)^{1/2}$-Lipschitz.
\end{lemma}
\begin{proof}
By Cauchy--Schwarz,

\begin{equation}
|F_j(U_1,\cdots, U_K)-F_j(U'_1,\cdots, U'_K)|= \left|\tr\parens{X^j-X'{}^j}\right|\leq\frac{\|X^j-X'{}^j\|_2}{\sqrt{d}}\,.
\end{equation}

Recalling $W = \Pi_K\dots\Pi_1$ and $X = W^\dag W$, we have the following chain of inequalities:
\begin{equation}
\|X^j-X'{}^j\|_2\leq j \|X-X'{}\|_2\leq 2j \|W-W'{}\|_2\leq 2j\sum_{i=1}^K\|\Pi_i-\Pi'_i\|_{2}\leq 2j\sum_{i=1}^{K}\|U_i-U'_i\|_2,
\end{equation}
where the first is using a telescoping sum, the second uses $\|ABC\|_2\leq \|A\|_{\infty}\|B\|_{2}\|C\|_\infty$ and that sandwiched projections are contractions, the third is again a telescoping trick; the last uses that, if $P$ is a projector, $\|U_iPU_{i}-U'_iPU'_{i}\|^2_2=\|[P,(U_iU_{i}'{}^{\dagger}-I)]]\|_2^2\leq \|U'_iU_{i}^{\dagger}-I\|^2_{2}=\|U_i-U'_{i}{}^{\dagger}\|^2_2$, with the inequality due to Hilbert-Schmidt orthogonality.
A further application of Cauchy--Schwarz gives
\begin{equation}
|F_j(U_1,\cdots, U_K)-F_j(U'_1,\cdots, U'_K)|\leq \left(\frac{4Kj^2}{d}\right)^{1/2}\sqrt{\sum_{i=1}^{K}\|U_i-U'_i\|^2_{2}}. \qedhere
\end{equation}
\end{proof}
We can then prove the main technical tool for proving spectral separation.
\begin{lemma}[Concentration of normalized power sums]
\label{lem:power-sum-concentration}
Fix a positive integer $\ell$. There is a constant $B>0$ such that, for every $\varepsilon>0$ and $e\in\{a,b\}$, the event
\begin{equation}
\wt{\mathcal E}_e(\varepsilon,\ell)
\coloneqq
\left\{
\|\brho_e\|_\infty>\frac{B}{d}
\quad\text{or}\quad
\max_{1\leq j\leq \ell}
\left|
\mathrm{pow}_j\parens{\spec(\brho_e)}
-\Psi_j(e)
\right|
>\varepsilon
\right\}
\label{eq:typical-spectrum-event}
\end{equation}
satisfies
\begin{equation}
\Pr\parens*{\wt{\mathcal E}_e(\varepsilon,\ell)}
\leq
\exp\parens{O(n)-\Omega(d^2)}.
\label{eq:atypical-spectrum-probability}
\end{equation}
The implicit constants may depend on $\varepsilon$ and $\ell$.
In particular, if $n=o(d^2)$, then
\begin{equation}
\Pr\parens*{\wt{\mathcal E}_a(\varepsilon,\ell)\,\cup\,\,  \wt{\mathcal E}_b(\varepsilon,\ell)}
\leq
\exp\parens{-\Omega(d^2)}. 
\end{equation} 
\end{lemma}
\begin{proof}
Let us first observe that the desired statement involves the random variable $\brho_e$, which is a function of $\wt{\bX}_e$. On the other hand, the expected moments calculated in~\Cref{sec:det_moment} are for the random variables $\bX_e$. The bridge is established by \Cref{prop:prob_tilt}, which reduces our problem to bounding $\nu_e^{-n}\Pr\parens*{\mathcal E_e(\varepsilon,\ell)}$, where
\begin{equation}
{\mathcal E}_e(\varepsilon,\ell)
\coloneqq
\left\{
\left\|\frac{\bX_e}{\Tr[\bX_e]}\right\|_\infty>\frac{B}{d}
\quad\text{or}\quad
\max_{1\leq j\leq \ell}
\left|
\mathrm{pow}_j\parens*{\spec\left(\frac{\bX_e}{\Tr[\bX_e]}\right)}
-\Psi_j(e)
\right|
>\varepsilon
\right\}\,.
\label{eq:atypical-untilted-event}
\end{equation}
From~\Cref{eq:average1}, $\nu_e=\Theta_k(1)$. Let $\nu_{*}=\min\{\nu_a,\nu_b\}>0$ and set $B=\frac{2}{\nu_*}$. Since $\bX_e$ is a contraction, $\left\|\frac{\bX_e}{\Tr[\bX_e]}\right\|_{\infty}\leq \frac{1}{d\tr(\bX_e)}$. Therefore,

\begin{equation}
\left\|\frac{\bX_e}{\Tr[\bX_e]}\right\|_{\infty}> \frac{2}{d\nu_*}\implies \tr(\bX_e)\leq \nu_e/2\implies |F_1(\bX_e)-\E F_1(\bX_e)|\geq \nu_e/2\,,
\end{equation}
 where we used the definition of $F_j$ in \Cref{lem:lipschitz}. Similarly, it is immediate that

\begin{equation}
\max_{1\leq j\leq \ell}
\left|
\mathrm{pow}_j\parens*{\spec\left(\frac{\bX_e}{\Tr[\bX_e]}\right)}
-\Psi_j(e)
\right|=\max_{1\leq j\leq \ell}
\left|\frac{F_j(\bX_e)}{F_1(\bX_e)^j}- \frac{\E [F_j(\bX_e)]}{\E[F_1(\bX_e)]^j}
\right|,
\end{equation}

We now use that the map $(x,y)\mapsto y/x^j$ is uniformly
Lipschitz when $x$ is bounded away from zero. More explicitly, since $\E F_1(\bX_e)=\nu_e=\Theta_{k}(1)$, there exists a constant $\eta_{\varepsilon,\ell}$ such that
\begin{equation}
|F_j(\bX_e)-\E F_j(\bX_e)|\leq \eta_{\varepsilon,\ell}\,\, \forall\, 1\leq j\leq \ell, \implies \max_{1\leq j\leq \ell}
\left|
\mathrm{pow}_j\parens*{\spec\left(\frac{\bX_e}{\Tr[\bX_e]}\right)}
-\Psi_j(e)
\right|\leq \varepsilon\,.
\end{equation}
Collectively, we have
\begin{equation}
\mathcal{E}_e(\varepsilon,\ell) \subseteq \left\{\max_{1\leq j\leq \ell}|F_j(\bX_e)-\E F_j(\bX_e)|\geq \min\left(\eta_{\varepsilon,\ell},\frac{\nu_{*}}{2}\right)\,\right\}.
\end{equation}
We can now use the previously established concentration tools for the larger event. 

Applying \Cref{eq:concentration-lips} and \Cref{lem:lipschitz} to each
$F_j$, and then taking a union bound, gives
\begin{align}
&\Pr\parens*{\mathcal{E}_e(\varepsilon,\ell)}\leq \Pr\left(\max_{1\leq j\leq \ell}|F_j(\bX_e)-\E F_j(\bX_e)|\geq \min\left(\eta_{\varepsilon,\ell},\frac{\nu_{*}}{2}\right)
\right)
\nonumber\\
&\qquad\leq
2\ell
\exp\left(
-\frac{
d(d-2)\min(\eta_{\varepsilon,\ell},\frac{\nu_{*}}{2})^2
}{
96K\ell^2
}
\right)
=2\ell\exp\left(-c_{\varepsilon,\ell}d(d-2)\right)
\label{eq:power-sum-concentration}
\end{align}

where $c_{\varepsilon,\ell}\coloneqq \frac{\min(\eta_{\varepsilon,\ell},\frac{\nu_{*}}{2})^2
}{
96K\ell^2}$. To conclude, using the tilting relation recalled above,

\begin{align}
\Pr\parens*{\wt{\mathcal E}_a(\varepsilon,\ell)\,\cup\,\,  \wt{\mathcal E}_b(\varepsilon,\ell)}&\leq \nu_{a}^{-n}\Pr\parens*{{\mathcal E}_a(\varepsilon,\ell)}+\nu_{b}^{-n}\Pr\parens*{{\mathcal E}_b(\varepsilon,\ell)}\\&\leq 4\ell\nu_{*}^{-n}\exp\left(-c_{\varepsilon,\ell}d(d-2)\right)\\&=4\ell\exp\left(-n\log\nu_{*}-c_{\varepsilon,\ell}d(d-2)\right),
\end{align}
proving the claim.
\end{proof}

We are ready to prove the spectral separation result:

\begin{proof}[Proof of \Cref{prop:typical-spectral-separation}]
By~\Cref{lem:first-spectral-gap},
\begin{equation}
\left|
\Psi_{k+1}(a)-\Psi_{k+1}(b)
\right|
=
|\kappa_k|+O(d^{-1}),
\qquad
\kappa_k\neq0.
\end{equation}
We can therefore choose a constant $\eta_k>0$ such that, for all
sufficiently large $d$,
\begin{equation}
\left|
\Psi_{k+1}(a)-\Psi_{k+1}(b)
\right|
\geq
4\eta_k.
\label{eq:proxy-gap-eta}
\end{equation}

Let $B$ be the constant from
\Cref{lem:power-sum-concentration}, applied with $\ell=k+1$ and
$\varepsilon=\eta_k$, and define
\begin{equation}
\mathcal T_e^{(d)}
\coloneqq
\left\{
\alpha\in\mathbb R_{\geq0}^d:
\sum_{i=1}^d\alpha_i=1,\ 
\|\alpha\|_\infty\leq\frac{B}{d},\
\max_{1\leq j\leq k+1}
\left|
\mathrm{pow}_j(\alpha)-\Psi_j(e)
\right|
\leq\eta_k
\right\}.
\label{eq:typical-spectrum-sets}
\end{equation}
By~\Cref{lem:power-sum-concentration},
\begin{equation}
\Pr\left(
\spec(\brho_e)\notin\mathcal T_e^{(d)}
\right)
\leq
\exp\parens{O(n)-\Omega(d^2)}
=
\exp\left(-\Omega_k(d^2)\right),
\end{equation}
where the last equality uses $n=o(d^2)$. This proves
\eqref{eq:typical-spectrum-probability}.

Now let $\alpha\in\mathcal T_a^{(d)}$ and
$\beta\in\mathcal T_b^{(d)}$. By the triangle inequality and
\eqref{eq:proxy-gap-eta},
\begin{align}
\left|
\mathrm{pow}_{k+1}(\alpha)-\mathrm{pow}_{k+1}(\beta)
\right|
&\geq
\left|
\Psi_{k+1}(a)-\Psi_{k+1}(b)
\right|
-
\left|
\mathrm{pow}_{k+1}(\alpha)-\Psi_{k+1}(a)
\right|
\nonumber\\
&\hspace{2cm}
-
\left|
\mathrm{pow}_{k+1}(\beta)-\Psi_{k+1}(b)
\right|
\nonumber\\
&\geq
2\eta_k.
\label{eq:typical-power-sum-gap}
\end{align}
Moreover,
\begin{equation}
\|\alpha\|_\infty,\|\beta\|_\infty
\leq
\frac{B}{d}.
\end{equation}
Applying~\Cref{prop:power-sum-to-TV} with $j=k+1$ gives
\begin{equation}
\dtvsorted{\alpha}{\beta}
\geq
\frac{
2\eta_k
}{
2(k+1)B^k
}
=
\frac{\eta_k}{(k+1)B^k}.
\end{equation}
The claim follows by setting
\begin{equation}
\delta_k^{\mathrm{spec}}
\coloneqq
\frac{\eta_k}{(k+1)B^k}>0.
\end{equation}
Finally, \eqref{eq:random-spectrum-separation} follows from
\eqref{eq:typical-spectrum-probability} and a union bound over
$e\in\{a,b\}$.
\end{proof}

\section{Entropy Separation}
\label{sec:entropy-separation}

A slight modification of the ensembles of the previous section allows us to establish an entropy separation as well. The main issue with the previous class is that the logarithm of the eigenvalues may be to small to control the entropy in terms of its moment expansion. The solution is to suitably depolarize the ensembles. For $p\in(0,1)$, let $\mathcal{N}_p$ denote the depolarizing
channel, acting on states as
\begin{equation}
\mathcal{N}_p(\rho)
=
p\rho+(1-p)\frac{\Id}{d}.
\end{equation} 

By data processing, the average $n$-copy states can only be less distinguishable, therefore we don't need to prove a new bound on the trace distance: for $n = o(d^{2 - \frac{4}{k+4}})$, and the two ensembles at fixed $k$,
\begin{align}
\Dtr{\E \mathcal{N}_p(\brho_a)^{\otimes n}}{\E \mathcal{N}_p(\brho_b)^{\otimes n}}&=  \Dtr{\mathcal{N}_p^{\otimes n}(\E\brho_a^{\otimes n})}{{\mathcal{N}_p^{\otimes n}(\E\brho_b^{\otimes n})} }\\&\leq\Dtr{\E \brho_a^{\otimes n}}{\E\brho_b^{\otimes n}}\\&=\Dtr{\ol{\rho}_a^{(n)}}{\ol{\rho}_b^{(n)}} = o(1).
\end{align}

In the following, it will be useful to have the following definition for the centered normalized moments of a state $\rho$ with eigenvalues
$\lambda_1,\ldots,\lambda_d$:
\begin{equation}
\mathrm{cm}_j(\rho)
\coloneqq
\frac{1}{d}
\sum_{i=1}^d
(d\lambda_i-1)^j.
\label{eq:centered-normalized-moments}
\end{equation}
With the convention $\mathrm{pow}_0(\rho)=1$, the binomial theorem
gives
\begin{equation}
\mathrm{cm}_j(\rho)
=
\sum_{\ell=0}^j
\binom{j}{\ell}
(-1)^{j-\ell}
\mathrm{pow}_\ell(\rho).
\label{eq:centered-moment-power-sum}
\end{equation}

The following expansion relates the centered moments to the entropy.

\begin{lemma}
\label{lem:entropy-expansion}
Suppose that $\|\rho\|_\infty\leq B/d$ and set
\begin{equation}
C_B\coloneqq\max\{1,B-1\}.
\end{equation}
If $pC_B<1$, then
\begin{equation}
\log d-S\left(\mathcal{N}_p(\rho)\right)
=
\sum_{j=2}^{\infty}
\frac{(-1)^jp^j}{j(j-1)}
\mathrm{cm}_j(\rho),
\label{eq:entropy-expansion}
\end{equation}
where the series converges absolutely and uniformly over all such
states.
\end{lemma}

\begin{proof}
The entropy of the depolarized state is

\begin{align}
S(\mathcal{N}_p(\rho))&=-\sum_{i=1}^d \left(p\lambda_i+\frac{(1-p)}{d}\right)\log \left(p\lambda_i+\frac{(1-p)}{d}\right)\\
&=\log d-\frac{1}{d}\sum_{i=1}^d \left(1+p(d\lambda_i-1)\right)\log \left(1+p(d\lambda_i-1)\right)\,.
\end{align}
Clearly, $d\lambda_i-1\geq -1$, while the assumption $\|\rho\|_{\infty}\leq \frac{B}{d}$ gives $d\lambda_i-1\leq B-1$. From the condition $pC_B<1$, it follows that $|p(d\lambda_i-1)|< 1$. For $|u|<1$,
\begin{equation}
(1+u)\log(1+u)
=
u+
\sum_{j=2}^{\infty}
\frac{(-1)^j}{j(j-1)}u^j.
\end{equation}
Therefore, summing over $i$ and rearranging
\begin{align}
S(\mathcal{N}_p(\rho))-\log d&=-\frac{1}{d}\sum_{i=1}^d \sum_{j=2}^{\infty}\frac{(-1)^j(p(d\lambda_i-1))^{j}}{j(j-1)}=-\sum_{j=2}^{\infty}\frac{(-1)^jp^j\mathrm{cm}_j(\rho)}{j(j-1)}\,
\end{align}
where the linear term cancels because $\sum_{i=1}^d p(d\lambda_i-1)=0$.
Uniform absolute
convergence follows from $p^j|\mathrm{cm}_j(\rho)|< (pC_B)^j$ with $pC_B<1$. 
\end{proof}

We can now prove the typical entropy separation.

\begin{proof}[Proof of \Cref{prop:typical-entropy-separation}]
Let $B_k$ be the operator-norm constant from
\Cref{lem:power-sum-concentration} and set
\begin{equation}
C_k'\coloneqq\max\{1,B_k-1\}.
\end{equation}
Choose $p_k>0$ sufficiently small that
\begin{equation}
p_kC_k'<1.
\end{equation}

We next choose a constant $\eta_k>0$, to be specified below, and
define
\begin{equation}
\mathcal U_e^{(d)}
\coloneqq
\left\{
\rho\in\mathcal D_d:
\|\rho\|_\infty\leq\frac{B_k}{d},\
\max_{1\leq j\leq k+1}
\left|
\mathrm{pow}_j(\rho)-\Psi_{e,j}
\right|
\leq\eta_k
\right\}.
\label{eq:typical-entropy-sets}
\end{equation}
By~\Cref{lem:power-sum-concentration},
\begin{align}
\Pr\left(
\brho_e\notin\mathcal U_e^{(d)}
\right)
&\leq
2(k+1)
\exp\left(
C_kn-c_{k,\eta_k}d(d-2)
\right)
\nonumber\\
&=
\exp\left(-\Omega_k(d^2)\right),
\label{eq:typical-entropy-set-probability}
\end{align}
where the last equality uses $n=o(d^2)$.

Now fix $\rho_a\in\mathcal U_a^{(d)}$ and
$\rho_b\in\mathcal U_b^{(d)}$. By
\Cref{eq:centered-moment-power-sum}, 
each centered moment is a fixed linear combination of the
power sums of equal or lower degree, and the coefficient of
$\mathrm{pow}_{k+1}$ in $\mathrm{cm}_{k+1}$ is one. Thus, by \Cref{lem:first-spectral-gap}, we have
\begin{align}
\mathrm{cm}_j(\rho_a)-\mathrm{cm}_j(\rho_b)
&=
O_k\left(\eta_k+d^{-1}\right),
\qquad
2\leq j\leq k,
\label{eq:centered-moment-matching}
\\
\mathrm{cm}_{k+1}(\rho_a)-\mathrm{cm}_{k+1}(\rho_b)
&=
\kappa_k
+
O_k\left(\eta_k+d^{-1}\right).
\label{eq:centered-moment-gap}
\end{align}
Now, applying~\Cref{lem:entropy-expansion} to $\rho_a$ and $\rho_b$ and
subtracting the two expansions gives
\begin{align}
&S\left(\mathcal{N}_{p_k}(\rho_a)\right)
-
S\left(\mathcal{N}_{p_k}(\rho_b)\right)
\nonumber\\
&\qquad=
\frac{
(-1)^{k+2}\kappa_kp_k^{k+1}
}{
k(k+1)
}
+
O_k\left(
\eta_k+d^{-1}+p_k^{k+2}
\right).
\label{eq:entropy-leading-term}
\end{align}
Here the first two error terms follow from
\eqref{eq:centered-moment-matching}--%
\eqref{eq:centered-moment-gap}. The last term bounds the remaining
tail, using
$|\mathrm{cm}_j(\rho_e)|< C'_{k}{}^j$ for any state in $\mathcal{U}_{a}^{(d)}$ or $\mathcal{U}_b^{(d)}$.

We can now complete the choices of constants. First choose $p_k>0$
sufficiently small that the $O_k(p_k^{k+2})$ term is at most one
quarter of the leading term. Then choose $\eta_k>0$ sufficiently
small that the $O_k(\eta_k)$ term is also at most one quarter of the
leading term. For all sufficiently large $d$, the $O_k(d^{-1})$
term satisfies the same bound. It follows from
\eqref{eq:entropy-leading-term} that
\begin{equation}
\left|
S\left(\mathcal{N}_{p_k}(\rho_a)\right)
-
S\left(\mathcal{N}_{p_k}(\rho_b)\right)
\right|
\geq
\frac{
|\kappa_k|p_k^{k+1}
}{
4k(k+1)
}.
\end{equation}
The claim follows by setting
\begin{equation}
\delta_k^{\mathrm{ent}}
\coloneqq
\frac{
|\kappa_k|p_k^{k+1}
}{
4k(k+1)
}
>0.
\end{equation}
Finally, \eqref{eq:random-entropy-separation} follows from
\eqref{eq:typical-entropy-probability} and a union bound over
$e\in\{a,b\}$.
\end{proof}

\section{Rank Separation}
\label{sec:rank-separation}
Throughout this section, we continue to fix $k$, $K$, $a$, and $b$ as in \Cref{prop:prouhet}. For most of the argument, we work with the untilted normalized states $\bX_e/\Tr(\bX_e)$, where $\bX_e\sim\PRP_d(e)$. At the end, we use \Cref{prop:prob_tilt} to transfer the resulting high-probability conclusion to $\brho_e=\wt{\bX}_e/\Tr(\wt{\bX}_e)$.

Label the two sequences $a,b$ so that $\max_i\{a_i\}=2^k$ and $\max_i\{b_i\}=2^k-1$. Set $r\coloneqq d/2^k$ and $\Delta\coloneqq 1/(2^k-1)-1/2^k$. Then $\operatorname{rank}(\brho_a)=r$ almost surely. Our goal is to show that $\brho_b$ is bounded away in trace distance from every state of rank at most $r$ with high probability.

\subsection{Limiting spectral mass analysis}

For an untilted draw $\bX_e$, let $\bL_{e,d}$ denote the empirical spectral law of $\bX_e/\Tr(\bX_e)$ after its eigenvalues are rescaled by a factor of $d$, so that the law is on an $O(1)$ scale.

\begin{lemma}[Limiting law and its atom at zero]
\label{lem:free-limit}
For each fixed parameter sequence $e\in\{a,b\}$, there is a compactly supported probability measure $\mu_e$ on $[0,\infty)$ such that, almost surely,
\begin{equation}
\label{eq:weak-convergence}
\bL_{e,d}\Longrightarrow\mu_e.
\end{equation}
Moreover,
\begin{equation}
\label{eq:zero-atom}
\mu_e(\{0\})=1-\frac{1}{\max_i\{e_i\}},
\end{equation}
and, for every fixed positive integer $j$,
\begin{equation}
\label{eq:Psi-limit}
\Psi_j(e)\longrightarrow\int_0^\infty t^j\,d\mu_e(t).
\end{equation}
\end{lemma}

\begin{proof}
Write the independent random projections defining $\bX_e$ as
\(
\bPi_i=\bU_iP_i^{(0)}\bU_i^\dagger,
\)
where the $\bU_i$ are independent Haar-random unitaries and $P_i^{(0)}$ is a deterministic projection of normalized rank $1/e_i$. Asymptotic freeness of independent Haar conjugates implies that $(\bPi_1,\ldots,\bPi_K)$ converges almost surely in noncommutative distribution to freely independent projections $(\bp_1,\ldots,\bp_K)$ with $\tau(\bp_i)=1/e_i$.
See, for example, Collins and Male~\cite{CollinsMale}.

Consequently, the empirical law of $\bX_e$ converges almost surely to the law of
\begin{equation}
\label{eq:free-product}
\bx_e\coloneqq \bp_1\bp_2\cdots \bp_K\cdots \bp_2\bp_1.
\end{equation}
Also,
\begin{equation}
\label{eq:trace-limit}
\tr(\bX_e)\longrightarrow\tau(\bx_e)=\prod_{i=1}^K\tau(\bp_i)\eqqcolon\nu_e>0.
\end{equation}
Indeed, if $\by$ belongs to the algebra generated by $\bp_{j+1},\ldots,\bp_K$, then traciality and freeness give \(\tau(\bp_j\by\bp_j)=\tau(\by\bp_j)=\tau(\by)\tau(\bp_j)\), and one iterates this identity. It follows that $\mu_e$ is the law of $\bx_e/\nu_e$, proving \Cref{eq:weak-convergence}.

Let
\begin{equation}
\label{eq:bernoulli-laws}
\beta_i\coloneqq\left(1-\frac{1}{e_i}\right)\delta_0+\frac{1}{e_i}\delta_1
\end{equation}
be the law of $\bp_i$. Recursively viewing \Cref{eq:free-product} as $\bp_j\by\bp_j=\bp_j^{1/2}\by\bp_j^{1/2}$ shows that
\begin{equation}
\label{eq:free-convolution}
\operatorname{Law}(\bx_e)=\beta_1\boxtimes\beta_2\boxtimes\cdots\boxtimes\beta_K.
\end{equation}
For probability measures on $[0,\infty)$, the atom formula for free multiplicative convolution is
\begin{equation}
\label{eq:atom-formula}
(\mu\boxtimes\nu)(\{0\})=\max\{\mu(\{0\}),\nu(\{0\})\};
\end{equation}
see Belinschi~\cite{BelinschiAtoms}. Iterating \Cref{eq:atom-formula} in \Cref{eq:free-convolution} gives $\operatorname{Law}(\bx_e)(\{0\})=1-1/\max_i\{e_i\}$. The positive rescaling $\bx_e\mapsto\bx_e/\nu_e$ does not change the atom at zero, proving \Cref{eq:zero-atom}. Factors with $e_i=1$ are identities and may simply be omitted.

Finally, for each fixed $j$, $0\leq\tr(\bX_e^j)\leq1$, so almost-sure moment convergence and bounded convergence yield
\begin{equation}
\label{eq:expected-moment-limit}
\E\tr(\bX_e^j)\longrightarrow\tau(\bx_e^j).
\end{equation}
Since $\E\tr(\bX_e)=\nu_e$ exactly, the definition of $\Psi_j(e)$ gives
\begin{equation*}
\Psi_j(e)=\frac{\E\tr(\bX_e^j)}{\nu_e^j}\longrightarrow\tau\parens*{(\bx_e/\nu_e)^j}=\int t^j\,d\mu_e(t),
\end{equation*}
which is \Cref{eq:Psi-limit}.
\end{proof}

For $e=b$, \Cref{lem:free-limit} gives
\begin{equation}
\label{eq:mu-b-zero}
\mu_b(\{0\})=1-2^{-k}-\Delta.
\end{equation}
Choose $c=c_k>0$ sufficiently small that
\begin{equation}
\label{eq:c-choice}
\mu_b([0,2c])\leq1-2^{-k}-\frac{3\Delta}{4}.
\end{equation}
Such a choice exists by continuity from above at the singleton $\{0\}$. Define a continuous cutoff $\phi:[0,\infty)\to[0,1]$ by
\begin{equation}
\label{eq:phi}
\phi(t)\coloneqq
\begin{cases}
1,&0\leq t\leq c,\\
2-t/c,&c<t<2c,\\
0,&t\geq2c.
\end{cases}
\end{equation}
Then
\begin{equation}
\label{eq:phi-properties}
\mathbbm{1}_{[0,c]}\leq\phi\leq\mathbbm{1}_{[0,2c]},
\qquad
\int\phi\,d\mu_b\leq1-2^{-k}-\frac{3\Delta}{4}.
\end{equation}

\subsection{Polynomial approximation and the eigenvalue count}

Let $B>0$ be the constant from \Cref{lem:power-sum-concentration}, and set $\alpha\coloneqq\Delta/16$. By the Weierstrass approximation theorem, there is a polynomial $P(t)=\sum_{j=0}^{\ell}\gamma_jt^j$ of some finite degree $\ell=\ell(k)$ such that
\begin{equation}
\label{eq:poly-approx}
\sup_{0\leq t\leq B}|P(t)-\phi(t)|\leq\alpha.
\end{equation}
After $k$ is fixed, the degree $\ell$ and all the coefficients $\gamma_j$ are constants independent of $d$ and $n$. Choose $\varepsilon>0$ sufficiently small that
\begin{equation}
\label{eq:epsilon-choice}
\varepsilon\sum_{j=1}^{\ell}|\gamma_j|\leq\alpha.
\end{equation}
By \Cref{eq:Psi-limit}, for all sufficiently large $d$,
\begin{equation}
\label{eq:center-convergence}
\left|\sum_{j=0}^{\ell}\gamma_j\Psi_j(b)-\int P\,d\mu_b\right|\leq\alpha,
\qquad
\Psi_0(b)\coloneqq1.
\end{equation}

Define the good event
\begin{equation}
\label{eq:good-event}
\mathcal G\coloneqq
\left\{\left\|\frac{\bX_b}{\Tr(\bX_b)}\right\|_\infty\leq\frac{B}{d}\right\}
\cap
\left\{\max_{1\leq j\leq\ell}\left|\mathrm{pow}_j\parens*{\frac{\bX_b}{\Tr(\bX_b)}}-\Psi_j(b)\right|\leq\varepsilon\right\}.
\end{equation}
The untilted concentration estimate in the proof of \Cref{lem:power-sum-concentration}, namely \Cref{eq:power-sum-concentration}, implies that, whenever $n=o(d^2)$, $\mathcal G$ occurs except with probability at most $\exp\parens{-\Omega(d^2)}$.

On $\mathcal G$, write $\by_i\coloneqq d\lambda_i\left(\bX_b/\Tr(\bX_b)\right)$, so $0\leq\by_i\leq B$, and set
\begin{equation}
\label{eq:LP}
L_P\left(\frac{\bX_b}{\Tr(\bX_b)}\right)\coloneqq\frac{1}{d}\sum_{i=1}^dP(\by_i)=\sum_{j=0}^{\ell}\gamma_j\mathrm{pow}_j\parens*{\frac{\bX_b}{\Tr(\bX_b)}},
\qquad
\mathrm{pow}_0\coloneqq1.
\end{equation}
The moment bound \Cref{eq:epsilon-choice} and the center convergence \Cref{eq:center-convergence} give the single estimate
\begin{equation}
\label{eq:LP-close}
\left|L_P\left(\frac{\bX_b}{\Tr(\bX_b)}\right)-\int P\,d\mu_b\right|\leq2\alpha.
\end{equation}
Meanwhile, \Cref{eq:phi-properties,eq:poly-approx} give
\begin{equation}
\label{eq:count-vs-LP}
\frac{1}{d}\#\{i:\by_i\leq c\}\leq L_P\left(\frac{\bX_b}{\Tr(\bX_b)}\right)+\alpha,
\qquad
\int P\,d\mu_b\leq\int\phi\,d\mu_b+\alpha.
\end{equation}
Combining only these two displays,
\begin{equation}
\label{eq:small-count}
\frac{1}{d}\#\{i:\by_i\leq c\}\leq1-2^{-k}-\frac{\Delta}{2}.
\end{equation}
Indeed, the total approximation loss is $4\alpha=\Delta/4$.

\subsection{Conversion to rank distance}

We will now prove \Cref{prop:rank-separation}.

\begin{proof}
We first prove the yes-case. If $\bW_e=\bPi_K\cdots\bPi_1$, then independent Haar-random subspaces are in general position almost surely, and therefore
\begin{equation}
\label{eq:generic-rank}
\operatorname{rank}(\bW_e)=\min_i\operatorname{rank}(\bPi_i)=\frac{d}{\max_i\{e_i\}}
\qquad\text{almost surely}.
\end{equation}
For completeness, condition on the first $j-1$ projections and let $\bS$ be the image of their product. An independent Haar-random kernel of dimension $d-d/e_j$ intersects $\bS$ in the smallest dimension allowed by dimension counting, almost surely; hence applying $\bPi_j$ reduces the rank to $\min\{\dim\bS,d/e_j\}$. Induction proves \Cref{eq:generic-rank}. Since $\operatorname{rank}(\bW_e^\dagger\bW_e)=\operatorname{rank}(\bW_e)$, the same is true of $\bX_e$. The tilted law is absolutely continuous with respect to the untilted law, and scalar normalization does not alter rank. Since $\max_i\{a_i\}=2^k$, this proves \Cref{eq:yes-rank}.

For the no-case, first work under the untilted law. On $\mathcal G$, enumerate the eigenvalues of $\bX_b/\Tr(\bX_b)$ in nonincreasing order. The bottom $(1-2^{-k})d$ eigenvalues contain, by \Cref{eq:small-count}, at least $\Delta d/2$ values satisfying $\by_i>c$. Consequently,
\begin{equation}
\label{eq:tail-mass}
\sum_{i>r}\lambda_i\left(\frac{\bX_b}{\Tr(\bX_b)}\right)=\frac{1}{d}\sum_{i>r}\by_i\geq\frac{c\Delta}{2} \eqqcolon \delta^{\mathrm{rank}}.
\end{equation}
By \Cref{eq:rank-approximation}, this tail is precisely the trace distance from $\bX_b/\Tr(\bX_b)$ to the nearest state of rank at most $r$. Thus \Cref{eq:tail-mass} proves the desired distance bound on $\mathcal G$ under the untilted law. The $n$-tilted law is absolutely continuous with density proportional to $\Tr(\bX_b)^n$, so \Cref{prop:prob_tilt} transfers the failure probability at a cost of at most $\nu_b^{-n}=\exp\parens{O(n)}$. Together with the preceding failure bound, this gives the overally required probability bound when $n=o(d^2)$.
\end{proof}

\bibliographystyle{alpha}
\bibliography{references} 

\appendix

\end{document}